\documentclass[a4paper,11pt]{article}
\usepackage{jheppub} 

\usepackage{lineno}

\usepackage{setspace}
\usepackage[normalem]{ulem}
\usepackage[utf8]{inputenc}

\usepackage[UKenglish]{babel}

\usepackage{textpos}
\usepackage{multicol}
\usepackage{graphicx}
\graphicspath{ {Imagenes/} }

\usepackage[mathscr]{euscript}

\usepackage{csquotes}
\usepackage{amsthm}
\usepackage{amsmath, bm}
\usepackage{amssymb}
\usepackage{tensor}
\usepackage{dsfont}

\usepackage{tikz-cd}

\newtheorem{Theorem}{Theorem} 

\newtheorem{Lemma}[Theorem]{Lemma}

\newtheorem{Definition}[Theorem]{Definition}
\newtheorem{Remark}[Theorem]{Remark}
\newtheorem{Example}[Theorem]{Example}

\title{Categorical generalization of BF theory coupled to gravity}

\author[a,c]{A.D. L\'opez-Hern\'andez,}
\author[b]{Graciela Reyes-Ahumada}
\author[c]{and Javier Chagoya}
\affiliation[a]{Departamento  de F\'{\i}sica, Universidad de Guanajuato, A.P. E-143, C.P. 37150, Le\'on, Guanajuato, M\'exico.}
\affiliation[b]{CONAHCyT-Unidad Académica de Matemáticas, Universidad Autónoma de Zacatecas, Calzada Solidaridad esquina Paseo a la Bufa S/N C.P. 
98060, M\'exico.}
\affiliation[c]{Unidad Acad\'emica de F\'isica, Universidad Aut\'onoma de Zacatecas,
98060, Zacatecas, M\'exico.}

\emailAdd{ad.lopez.hernandez@ugto.mx, grace@cimat.mx, javier.chagoya@fisica.uaz.edu.mx}

\abstract{We present a thorough introduction to the tools of category theory required for formulating gauge theories based on 2-connections. We provide a detailed construction of the categorical generalization of BF theory, dubbed BFCG, also known as 2BF. Similar to BF gravity, it is known that BFCG can be deformed to give general relativity. Here, we obtain an alternative relation between BFCG and gravity, which consists of coupling general relativity and BFCG by means of the volume form constructed out of the BFCG connections. The resulting theory, closely related to unimodular gravity, is a generalization of BF sequestered gravity not only in the sense that it adds new fields but also in that it allows for new choices for the volume form that is coupled to gravity. Furthermore, we show that BF sequestered gravity in the abelian case is recovered for a specific choice of the 2-group.}

\begin{document}
\maketitle
\flushbottom

\section{Introduction}
\label{sec:intro}
Symmetry principles have played a fundamental role in the construction of physical theories. The first example of a gauge field theory that successfully combined quantum mechanics and special relativity is quantum electrodynamics (QED): an abelian gauge theory for $U(1)$ that describes the interactions between photons, electrons and positrons. The generalization from abelian to non-abelian gauge theories led to Yang-Mills theories and ultimately to the development of the standard model of particle physics (SMPP), which is a spontaneously broken non-abelian gauge theory based on the symmetry group $SU(3)\times SU(2)\times U(1)$, describing the strong force, the electromagnetic force and the weak force. These are three of the four forces considered as fundamental, the other one being gravity. Essentially, the fact that gravity is described by a classical field theory is the reason why it is not included in the SMPP, which is a quantum field theory. Efforts to address this shortcoming include attempts to formulate a theory of gravity in the language of Yang-Mills theories, even if it is not a YM theory in the traditional sense. For instance, recent developments have related gravity to (YM)${}^2$ theories~\cite{Bern:2010yg}. Another possibility is to write gravity as a diffeomorphism invariant YM theory~\cite{Krasnov:2011up} or as deformations of topological field theories, specifically of $BF$ theories~\cite{Celada_2016}. Deformations, including constraints and potential terms added to the topological $BF$ action, can transform some of the gauge degrees of freedom in the topological $BF$ action into physical ones. Examples include the Plebanski formulation of General Relativity (GR), which imposes the simplicity constraint on the $B$ field, and the $BF$ formulation of Yang-Mills theories, which is realized by adding potential terms~\cite{Fucito:1996ax}. Another notable example of a deformation involving potential terms is the $BF$ reformulation of MacDowell-Mansouri gravity introduced by Freidel and Starodubtsev~\cite{Freidel:2005ak}

In the field of modified gravity, a new relationship between $BF$ theories and gravity has recently been proposed under the name of \textit{$BF$ sequestered gravity}, in which Lagrangians for both theories are coupled such that the volume form of spacetime must coincide with the volume form constructed from one of the fields of the $BF$ theory, with this volume form acting as the potential term~\cite{alexander2020topological}. One result of this coupling is that the field equations include those of an alternative theory to GR, known as Unimodular Gravity (UG)~\cite{1952prel.book..189E, Unruh:1988in}. There are various ways to motivate UG, but its central point is that the observed cosmological constant is attributed to an integration constant that arises from manipulating the field equations. This is in contrast to General Relativity, where the observed cosmological constant is tied to the vacuum energy density. One way to obtain UG from a principle of least action is by considering theories with Weyl invariance~\cite{Alvarez:2015sba}. This is where the connection with the Einstein-Hilbert theory coupled to $BF$ comes in, as the previously mentioned compatibility condition between the volume forms also leads to Weyl invariance. 

On the other hand, there are generalizations of gauge theories, known as higher gauge theories~\cite{baez2011invitation}. In general terms, these generalizations introduce a 2-form connection, in addition to the 1-form connection of conventional gauge theory. While the 1-form connection provides curves with holonomies in a gauge group, the 2-form connection is used to provide surfaces with a new type of surface holonomy, represented by elements of another group. Using higher gauge theories, it is possible to generalize the $BF$ theory, leading to formulations of topological field theories~\cite{Mikovic:2016xmo}, for instance, $BFCG$ (also known as $2BF$) for 2-categories and $3BF$ for 3-categories (see~\cite{girelli2008topological, Radenkovic:2019qme}). It has been shown that these categorical generalizations of $BF$ allow to couple gauge and matter fields to GR. The dynamical degrees of freedom appear after imposing a simplicity constrain~\cite{Radenkovic:2019qme}, similar to the deformation of $BF$ theory that leads to a formulation of GR.

The main purpose of this work is to offer a pedagogical, mostly self-contained introduction to the tools of category theory required for generalizing gauge theories, and to provide a physical application given by an alternative coupling between $BFCG$ and gravity along the lines of $BF$ sequestered gravity, providing a new application of higher gauge theories. As we show in the following sections, this methodology leads to a theory that couples UG with the gauge fields of $BFCG$ theory.

The structure of this work is as follows: In section~\ref{sec:pre}, we introduce the concept of $G$-principal bundles along with examples of their physical applications. In section~\ref{sec:cat}, we categorize the principal bundles and subsequently generalize them to introduce the notion of 2-principal bundles and their associated 2-form connections. In section~\ref{sec:phys}, we present a generalized coupling (original work intended for this article) between gravity and $BF$ theories, and we propose a coupling with a kinetic term. Section~\ref{sec:con} is devoted to conclusions. 

\section{Preliminaries}\label{sec:pre}

\subsection{Principal bundles and connections}
In this section we recall standard notation from differential geometry such as principal $G$-bundles, connections, horizontal lift and holonomy, as in~\cite{nakahara2003geometry,michor2008topics}. We also refer to classical texts such as \cite{nomizu1963foundations}. Throughout this paper we consider real differential manifolds and  differentiable morphisms. We let $G$ be a compact connected Lie group.

\begin{Definition}
 Let $M$ be a smooth manifold.   
 A $G$-fibration on $M$ is a morphism $\pi:P\rightarrow M$  of smooth manifolds together with a differentiable right action $\triangleleft :P\times G\rightarrow P$ and a $G$-equivariant morphism $\pi (p\triangleleft g)=\pi (p), \forall g\in G.$
 A morphism between two $G$-fibrations $\pi:P\rightarrow M$ and $\pi ^{\prime}:P^{\prime}\rightarrow M$ is a morphism $f:P\rightarrow P^{\prime}$ such that $\pi=\pi'\circ f$ (equivariant).
$P$ and $P'$ are isomorphic $G$-fibrations if there are equivariant morphisms  $f:P\rightarrow P^{\prime}$ and $g:P'\rightarrow P$ with $g\circ f=\text{id}_{P}$ and $f\circ g=\text{id}_{P'}$.
A $G$-fibration is trivial when it is isomorphic to the $G$-fibration projection $pr_{1}:M\times G\rightarrow M$ with the right action 
$$ \rho:  (M\times G)\times G  \rightarrow M\times G,\; \; \; 
        ((m,g),g^{\prime}) \mapsto (m,gg^{\prime}).$$
\end{Definition}
 
\begin{Definition}
    A principal $G$-bundle is a $G$-fibration which is locally trivial, that, is for any point $m\in M$ there is a neighborhood $U\ni m$ such that $\pi|_{\pi ^{-1}(U)}:\pi ^{-1}(U)\rightarrow U$ is trivial.
    A section of a principal $G$-bundle $\pi :P\rightarrow M$ over $U\subset M$ is a morphism $s:U\rightarrow P$ such that $\pi \circ s=\text{id}_{U}$. A global section is a section $s:M\rightarrow P$.
\end{Definition}

\noindent It is well known that a principal $G$-bundle is trivial if there exists a global section, and that there exists an open cover $\{U_{i}\}$ of $M$ with sections $\sigma _{i}$ defined on each $U_{i}$.

    \begin{Definition}
        Let $P$ be a smooth manifold, and $G$ a Lie group with Lie algebra $\mathfrak{g}\equiv T_{e}G$. A $\mathfrak{g}$-valued $n$-form over the manifold $P$ is an element of $\Omega ^{n}(P,\mathfrak{g})=\Omega ^{n}(P)\otimes \mathfrak{g}$, where $\Omega ^{n}(P)$ is the space of all $n$-forms on $P$.
    \end{Definition}
    
\begin{Remark}
    If $\omega \in \Omega ^{1}(P,\mathfrak{g})$ is a $\mathfrak{g}$-valued 1-form then $\omega$ is a $C^{\infty}(P)$-linear function $\omega: \mathfrak{X}(P)\rightarrow C^{\infty}(P,\mathfrak{g})$, where $\mathfrak{X}(P)$ is the set of all vector fields on $P$. Locally, 
    \begin{equation*}
        \omega = \omega _{\mu}^{a}dx^{\mu}\otimes g_{a},
    \end{equation*}

    \noindent where $\{g_{a}\}$ is a basis for $\mathfrak{g}$. Thus, for any vector field $X\in \mathfrak{X}(P)$,
    \begin{equation*}
        \omega (X)=(\omega _{\mu}^{a}dx^{\mu}X)g_{a},
    \end{equation*}

    \noindent where $(\omega _{\mu}^{a}dx^{\mu}X)\in C^{\infty}(P)$.
\end{Remark}

\begin{Definition} 
    Let $\pi :P\rightarrow M$ be a principal $G$-bundle. A 1-form connection $\omega$ is a $\mathfrak{g}$-valued 1-form over $P$, satisfying the following properties:

    \begin{enumerate}
    \renewcommand{\theenumi}{\alph{enumi}}

    \item For any point $p\in P$ and any element $A\in \mathfrak{g}$, it holds that $\omega_{p}(X_{p}^{A})=A$,
    \begin{equation*}
        \omega(X^{A})=A, \;\text{for any}\; A\in \mathfrak{g},
    \end{equation*}
    where the field $X^{A}$ represents the fundamental vector field induced by $A$, defined as
    \begin{equation}
        X_{p}^{A}f = \left(\frac{d}{dt}f(p\triangleleft \text{exp}(tA))\right)\bigg|_{t=0},
    \end{equation}
    for any $p\in P$ and $f\in C^{\infty}(P)$.
    
    \item For any $p\in P$, $g\in G$ and any tangent vector $X_{p}\in T_{p}P$, we get $((\triangleleft g)^{\ast}\omega)_{p}(X_{p})=(Ad_{g^{-1}})_{\ast}(\omega_{p}(X_{p}))$,
    \begin{equation*}
        (\triangleleft g)^{\ast}\omega=(Ad_{g^{-1}})_{\ast}\omega, \;\text{for any}\;\; g\in G.
    \end{equation*}
where $\triangleleft g:P \rightarrow P,\; \; \;
                p  \mapsto p\triangleleft g.
          $
    \end{enumerate}
\end{Definition}

\begin{Definition}
    Let $\pi :P\rightarrow M$ be a principal $G$-bundle. The \emph{vertical space} at a point $p\in P$ is defined as
    \begin{equation*}
        V_{p}P=\{X\in T_{p}P \mid \pi_{\ast} (X)=0\}\subset T_{p}P,
    \end{equation*}
    where $\pi _{\ast}$ is the push-forward or differentiation of $\pi$ at $p$.
\end{Definition}
\begin{Lemma}
    Let $A\in \mathfrak{g}$. Then $X_{p}^{A}\in V_{p}P$ for any $p\in P$.
\end{Lemma}

\begin{proof}
   $ \pi _{\ast}(X_{p}^{A})f=X_{p}^{A}(\pi ^{\ast}f)\equiv X_{p}^{A}(f\circ \pi)=(f\circ \pi (p \triangleleft \textup{exp}(tA)))^{\prime}(0)=(f\circ \pi (p))^{\prime}(0)=0,
   $
    for any $f\in C^{\infty}(M)$. Therefore $\pi _{\ast}(X_{p}^{A})=0$.
\end{proof}

\begin{Lemma}
    For any $p\in P$ the map
    $$        X_{p}^{(\cdot)}:\mathfrak{g} \rightarrow V_{p}P ,\; \; \;
        A  \mapsto X_{p}^{A}
      $$
    is an isomorphism of vector spaces. Thus, for every vertical vector field $X$ there exists a unique element $A\in \mathfrak{g}$ such that $X=X^{A}$.
\end{Lemma}

\begin{Theorem}
    Let $\omega$ be a 1-form connection, $p\in P$ and consider the horizontal space 
    \begin{equation}
        H_{p}P=\text{Ker}(\omega _{p})=\{X\in T_{p}P \mid \omega_{p}(X)=0\}.
    \end{equation}
     Then for all $p\in P$, the following properties hold:

    \begin{enumerate}
    \renewcommand{\theenumi}{\alph{enumi}}
    \item $H_{p}P\oplus V_{p}P=T_{p}P$.
    \item $(\triangleleft g)_{\ast}(H_{p}P)=H_{p\triangleleft g}P$ for any $g\in G$.
    \item The unique decomposition $X_{p}=\text{hor}(X_{p})+\text{ver}(X_{p})$ into its horizontal and vertical components maps every vector field $X$ to a decomposition into two vector fields, $\text{hor}(X)$ and $\text{ver}(X)$.
    \end{enumerate}

     \begin{Theorem}
        Let $H_{p}P$ be a vector subspace of $T_{p}P$ for all $p\in P$ satisfying the three properties of a horizontal space.  
    Then the 1-form $\omega$ defined as 
   $$            \omega _{p}:T_{p}P  \rightarrow \mathfrak{g},\; \; \;
            Y_{p}  \mapsto \omega _{p}(Y_{p})=(X_{p}^{(\cdot)})^{-1}(ver(Y_{p})),    $$
    for any $p\in P$ is a 1-form connection, and it satisfies $ker(\omega _{p})=H_{p}P$.
     \end{Theorem}
\end{Theorem}

\begin{Definition}
    The Maurer-Cartan 1-form $\Theta$ is defined as
    \begin{equation}
        \begin{split}
            \Theta _{g}:T_{g}G & \rightarrow T_{e}G\equiv \mathfrak{g}\\
            X_{g} & \rightarrow (L_{g^{-1}})_{\ast}(X_{g})
        \end{split}
    \end{equation}
    for any $g\in G$, where $L_{g}:G  \rightarrow G,\; \; \;
            h  \rightarrow gh,
  $
    is the left translation respect to $g$.

\end{Definition}

\begin{Definition}
    Let $\pi: P\rightarrow M$ be a principal $G$-bundle, and $\omega$ be a 1-form connection. Choose a local section $\sigma: U\rightarrow P$, where $U\subset M$, also known as a gauge (local). The local Yang-Mills potential (or local connection) is the $\mathfrak{g}$-valued 1-form $A$ (defined on $U\subset M$), given by,
    \begin{equation}
        A=\sigma ^{\ast}\omega.
    \end{equation}

\end{Definition}

\begin{Definition}
    Let $\sigma_{1}: U_{1}\rightarrow P$ and $\sigma_{2}: U_{2}\rightarrow P$ be local sections such that $U_{1}\cap U_{2} \neq \emptyset$. We refer to the local gauge transformation as the function
          $$  g:U_{1}\cap U_{2}  \rightarrow G, \; \; \; m  \mapsto g(m),$$
    where $g(m)$, for any $m\in U_{1}\cap U_{2}$, is defined as the unique element of $G$ such that
    \begin{equation*}
        \sigma _{2}(m)=\sigma_{1}(m)\triangleleft g(m)\equiv (\triangleleft g \circ \sigma _{1})(m),
    \end{equation*}
     where
   $$\triangleleft g:\pi ^{-1}(U_{1}\cap U_{2})  \rightarrow P,\; \; \;
             p  \rightarrow p\triangleleft g(\pi (p)).$$\\
    Such element always exists since $\sigma _{1}(m)$ and $\sigma _{2}(m)$ are in the same fiber.

\end{Definition}

\begin{Remark}
     Let $g: M \rightarrow G$ be a smooth function, $\sigma_{1}: U \rightarrow P$ a local section and $g _{U}: U \rightarrow P$ the restriction of $g$ to the subset $U$. Since $\pi$ is $G$-invariant, the function $\sigma_{2}: U \rightarrow P$ defined as $\sigma_{2}(m)=\sigma_{1}(m)\triangleleft g_{U}(m)$ is also a local section. Thus, any smooth function $g: M \rightarrow G$ will be referred to as a gauge transformation.
\end{Remark}

\begin{Theorem}
     Let $\sigma _{1}:U_{1}\rightarrow P$, $\sigma _{2}:U_{2}\rightarrow P$ be local sections, $\omega$ a 1-form connection, and $g:U_{1}\cap U_{2}\rightarrow G$ the local gauge transformation. Then, 
     \begin{equation}
        A_{2}=(Ad_{g^{-1}})_{\ast}A_{1}+g^{\ast}\Theta,
    \end{equation}

    \noindent where $A_{1}=\sigma _{1}^{\ast}\omega$ and $A_{2}=\sigma _{2}^{\ast}\omega=(\triangleleft g \circ \sigma _{1})^{\ast}\omega$. If $G\subset GL(n)$ the above identity simplifies to
    \begin{equation}
       A_{2}=g^{-1}A_{1}g+g^{-1}dg,
    \end{equation}
    where $g^{-1}:U_{1}\cap U_{2}\rightarrow G$ is defined as $g^{-1}(m)=(g(m))^{-1}\in G$ for any $m\in U_{1}\cap U_{2}$, that is, the inverse of $g$ at each point, and $dg$ is the \emph{exterior derivative} of $g$.

\end{Theorem}

\begin{Remark}
    \noindent In a chart $(U,x)$ in $M$, we can express the local components of $g^{-1}dg$ as
    \begin{equation*}
        (g^{-1}dg)_{j}^{i}=(g^{-1})_{k}^{i}(\frac{\partial}{\partial x^{\mu}}g_{j}^{k})dx^{\mu},
    \end{equation*}
that is,
\begin{equation*}
        g^{-1}dg=g^{-1}\partial_{\mu}g\; dx^{\mu}.
    \end{equation*}
\noindent If we consider a principal $G$-bundle equipped with a left action, then the gauge transformation of the connection is written as
\begin{equation*}
        A_{2}=g^{-1}A_{1}g-g^{-1}dg,
    \end{equation*}
or equivalently, $A_{1}=gA_{2}g^{-1}-gdg^{-1}$, which can be addressed by making the change.
\begin{equation}
    A\longrightarrow -A,
\end{equation}
in
\begin{equation}
    A_{2}=g^{-1}A_{1}g+g^{-1}dg,
\end{equation}
or equivalently, $A_{1}=gA_{2}g^{-1}+gdg^{-1}$. The preceding expressions are typically written as elements of the Lie algebra of the Lie Group $G$. Alternatively, if we express the connection in terms of the generator as usual, where we make the change
\begin{equation}
    A\longrightarrow-iA,
\end{equation}
then the gauge transformation of the connection is written as
\begin{equation}
    A_{2}=g^{-1}A_{1}g+ig^{-1}dg,
\end{equation}
or equivalently, $A_{1}=gA_{2}g^{-1}+igdg^{-1}$.
\end{Remark}

\begin{Theorem}
    Let $\pi: P \to M$ be a principal $G$-bundle equipped with a 1-form connection $\omega$, and let $\gamma: [0,1] \to M$ be a smooth curve with $\gamma(0) = x_{i}$ the initial point. Choose a point $p \in \pi^{-1}(x_{i})$ in the fiber over $x_{i}$. Then, there exists a unique smooth curve $\gamma^\uparrow: [0,1] \to P$ passing through the point $p$ (i.e., $\gamma^\uparrow(0) = p$) and satisfying the following conditions:
    \begin{enumerate}
        \item $\pi \circ \gamma ^{\uparrow}=\gamma$, meaning it lies in the fibers of the points that $\gamma$ passes through.
        \item $\omega _{\gamma ^{\uparrow}(t)}(X_{\gamma ^{\uparrow},\gamma ^{\uparrow}(t)})=0$, for any $t \in [0,1]$, indicating it is horizontal.

    \end{enumerate}
This curve is called the \emph{horizontal lift} of $\gamma$ through the point $p$.
\end{Theorem}
\noindent The construction of the horizontal lift leads to important concepts within gauge theories, such as parallel transport and holonomies. Locally, let $\sigma :U\rightarrow P$ be a local section such that $\gamma (0)\in U$, and let $g_{0}\in G$ be the group element such that $p=\sigma (\gamma (0)) \triangleleft g_{0}$. Then, we can express $\gamma ^{\uparrow}(t)=\sigma (\gamma (t)) \triangleleft g(\gamma (t))$, where $g:U \rightarrow G$, with the initial condition $g(\gamma (0))=g_{0}$ is such that
\begin{equation}
\label{eqdifhol}
    dg_{\gamma (t)}X_{\gamma,\gamma (t) }=-(Ag)_{\gamma (t)}X_{\gamma,\gamma (t) },
\end{equation}
where $A=\sigma ^{\ast}\omega $. This can be obtained from the condition of horizontability  
\begin{equation*}
    0=\omega _{\gamma ^{\uparrow}(t)}(X_{\gamma ^{\uparrow},\gamma ^{\uparrow}(t)}),
\end{equation*}
along with expression
\begin{equation*}
    \gamma ^{\uparrow}=(\triangleleft  g\circ \sigma)\circ \gamma .
\end{equation*}
and applying the gauge transformation.\\
Simplifying the left-hand side of eq.~\eqref{eqdifhol},
\begin{equation*}
    dg_{\gamma (t)}X_{\gamma,\gamma (t) }=X_{\gamma,\gamma (t) }g=\frac{d}{dt}(g\circ \gamma)(t),
\end{equation*}
and the right-hand side,
\begin{equation*}
   (Ag)_{\gamma (t)}X_{\gamma,\gamma (t) }=A_{\mu}gdx^{\mu}X_{\gamma,\gamma (t) }=A_{\mu}gX_{\gamma,\gamma (t) }x^{\mu}=A_{\mu}(\gamma (t))\dot{\gamma}^{\mu}(t)g(\gamma (t)),
\end{equation*}
we obtain,
\begin{equation}
\label{eqdifhol2}
    \frac{d}{dt}(g\circ \gamma)(t)=-A_{\mu}(\gamma (t))\dot{\gamma}^{\mu}(t)(g\circ \gamma)(t).
\end{equation}
The solution to the differential equation \eqref{eqdifhol2} for $t\in \gamma ^{-1}(U)$ with the initial condition $g(\gamma (0))=g_{0}$ is
\begin{equation}
    g(\gamma (t))= \mathcal{P}\textup{exp} (\int _{\gamma _{[0,t]}}-A)g_{0},
\end{equation}
where $\gamma _{[0,t]}$ is the restriction of $\gamma$ to the set $[0,t]$ and 
\begin{equation*}
\begin{split}
    \mathcal{P}\textup{exp} (\int _{\gamma _{[0,t]}}-A)=[1&-\int _{0}^{t}d\lambda _{1}A_{\mu}(\gamma (\lambda _{1}))\dot{\gamma}^{\mu}(\lambda _{1})\\
    &+\int _{0}^{t}d\lambda _{1}\int _{0}^{\lambda _{1}}d\lambda _{2}A_{\mu}(\gamma (\lambda _{1}))\dot{\gamma}^{\mu}(\lambda _{1})A_{\mu}(\gamma (\lambda _{2}))\dot{\gamma}^{\mu}(\lambda _{2})-  ...],
\end{split}
\end{equation*}
is the path-ordered exponential.
\noindent The (local) uniqueness follows from the uniqueness of the solution to a differential equation with initial condition. To construct the complete horizontal lift (not only locally), we proceed as follows. Since the interval $[0,1]$ is compact and connected, the image of $\gamma$ is also compact and connected. Therefore, there exists a finite cover composed of trivializations (or gauges) ${(U_i, \sigma_i)}_{i=1,\ldots,n}$, where $\sigma_i$ for $i=1,\ldots,n$ are local sections. The horizontal lift is constructed as follows: 

\noindent Choose an $i\in {1,\ldots,n}$ such that $\gamma(0) \in U_i$. As above, the horizontal lift can be written locally as
\begin{equation}
    \gamma ^{\uparrow}_{i}(t)=\sigma _{i} (\gamma (t))\triangleleft \mathcal{P}\textup{exp} (\int _{\gamma _{[0,t]}}-A_{i})g_{0} ,
\end{equation}
where $A_{i}=\sigma _{i}^{\ast} \omega$ and $t\in \gamma ^{-1}(U_{i})$. If $\gamma [0,1]\subset U_{i}$ then we are done. Otherwise, choose a $j\neq i$ and $t_{ij}\in [0,1]$ such that $\gamma (t_{ij})\in U_i \cap U_j$ (connected). Then, the horizontal lift for $t\in \gamma ^{-1}(U_{j})\cap [t_{ij},1]$ is
\begin{equation}
    \gamma ^{\uparrow}_{j}(t)=\sigma _{j}( \gamma(t))\triangleleft \mathcal{P}\textup{exp} (\int _{\gamma _{[t_{ij},t]}}-A_{j})g_{1} ,
\end{equation}
where $g_{1}$ (initial condition) is such that
\begin{equation*}
    \sigma _{j}(\gamma (t_{ij}))\triangleleft g_{1}=\gamma ^{\uparrow}_{i}(t_{ij}).
\end{equation*}
If $\gamma [0,1]\subset U_{i}\cup U_{j}$ then we are done. Otherwise, we continue as above.\\

\begin{Remark}
The expression
\begin{equation}
    \mathcal{P}\textup{exp} (\int _{\gamma }-A),
\end{equation}
is called holonomy and has interesting properties as a gauge transformation
\begin{equation}
    \mathcal{P}\textup{exp} (\int _{\gamma}-A_{2})=g(x_{f})^{-1}\mathcal{P}\textup{exp} (\int _{\gamma }-A_{1})g(x_{i}),
\end{equation}
where $x_{i}$ and $x_{f}$ are the initial and final points, respectively, of $\gamma$ and $g$ is the gauge transformation between the local connections. This relation can be derived from the uniqueness of the horizontal lift and the gauge relations between the corresponding sections.\end{Remark}

\noindent Consider a principal $G$-bundle equipped with a left action, the holonomy is 
\begin{equation}
    \mathcal{P}\textup{exp} (\int _{\gamma }A),
\end{equation}
which can be addressed  by making the change
\begin{equation}
    A\longrightarrow -A.
\end{equation}
Alternatively, if we express the connection in terms of generators as usual, where we make the change
\begin{equation*}
    A\longrightarrow-iA,
\end{equation*}
the holonomy becomes
\begin{equation}
    \mathcal{P}\textup{exp} (i \int _{\gamma }A),
\end{equation}

\noindent which is also referred to as the Wilson line.

\begin{Definition}
    The curvature of the 1-form connection $\omega$ is a $\mathfrak{g}$-valued 2-form defined as
    \begin{equation}
        \Omega (X_{1},X_{2})=d\omega (hor(X_{1}),hor(X_{2})), 
    \end{equation}
    for any $X_{1}, X_{2} \in \mathfrak{X}(P)$.
   The exterior covariant derivative $d_{\omega}\phi$ of a $\mathfrak{g}$-valued $(n-1)$-form  $\phi$ is defined as
    \begin{equation*}
        D\phi =d\phi \circ hor_{n},
    \end{equation*}
    where $hor_{n}:TP^{n}\rightarrow HP^{n}$ is the horizontal projection on each component.\end{Definition}
    
\begin{Theorem}
\label{curvatura}
The curvature can be expressed as
    \begin{equation}
        \Omega =d\omega + \omega \wedge \omega,
    \end{equation}
   where
   \begin{equation}
        \omega \wedge \omega (X_{1},X_{2})=[\omega (X_{1}),\omega (X_{2})],
    \end{equation}
   for any $X_{1}, X_{2} \in \mathfrak{X}(P)$, with $[\ ]$ denoting the commutator in the Lie algebra $\mathfrak{g}$.
\end{Theorem}

\begin{Theorem}
    The curvature satisfies the Bianchi identity
    \begin{equation}
        d_{\omega}\Omega =0.
    \end{equation}
  
\end{Theorem}
\begin{proof}
From the identity $\omega (X)=0$ for all $X\in HP$ and $d^{2}=0$, 
      $$   d_{\omega}\Omega=d(d\omega + \omega \wedge \omega)\circ hor_{3}=(d(\omega) \wedge \omega-\omega \wedge d( \omega))\circ hor_{3}=0. $$\end{proof}
  
    \begin{Lemma}
    Let $\{g_{a}\}$ be a basis for $\mathfrak{g}$ and let $\omega=\omega ^{a}\otimes g_{a} \in \Omega^{1}(P,\mathfrak{g})$ be a $\mathfrak{g}$-valued 1-form. Then,
    \begin{equation}
        \omega \wedge \omega =\frac{1}{2}[\omega \wedge \omega],
    \end{equation}
    where
    \begin{equation}
        [\omega \wedge \omega]=\omega ^{a} \wedge \omega ^{b}\otimes [g_{a},g_{b}].
    \end{equation}
\end{Lemma}

\begin{proof}
    Let $X_{1}, X_{2} \in \Gamma (TP)$ be vector fields. Then,
     
    \begin{equation*}
        \begin{split}
            \frac{1}{2}[\omega \wedge \omega](X_{1}, X_{2}) & =\frac{1}{2}\omega ^{a} \wedge \omega ^{b}(X_{1}, X_{2})[g_{a},g_{b}]=\frac{1}{2}(\omega ^{a}(X_{1})\omega ^{b}(X_{2})-\omega ^{a}(X_{2})\omega ^{b}(X_{1}))[g_{a},g_{b}]\\[0.2cm]
            & =\omega ^{a}(X_{1})\omega ^{b}(X_{2})[g_{a},g_{b}]=[\omega ^{a}(X_{1})g_{a},\omega ^{b}(X_{2})g_{b}]\\[0.2cm]
            & =[\omega (X_{1}),\omega (X_{2})]=\omega \wedge \omega (X_{1},X_{2}).
        \end{split}
    \end{equation*}

\end{proof}

\begin{Theorem}
    Let $\pi :P\rightarrow M$ be a principal $G$-bundle, with a 1-form connection $\omega$ and a local section (gauge) $\sigma :U\rightarrow P$. Then, the local curvature $F$ defined as $\sigma ^{\ast} \Omega$ satisfies
    \begin{equation}
        F=dA+A\wedge A, 
        \label{local curvature}
    \end{equation}
    where $A=\sigma ^{\ast}\omega$ is the Yang-Mills potential (or local connection). In components,
    \begin{equation}
        F=\frac{1}{2}F_{\mu \nu}dx^{\mu}\wedge dx^{\nu}, \;\;\; A=A_{\mu}dx^{\mu},
    \end{equation}
    where $F_{\mu \nu}$ and $A_{\mu}$ are valued in the Lie algebra of the Lie group $G$. With this eq.~\eqref{local curvature} is rewritten as 
    \begin{equation}
        F_{\mu \nu}=\partial _{\mu}A_{\nu}-\partial _{\nu}A_{\mu}+[A_{\mu},A_{\nu}].
        \label{comp local curvature}
    \end{equation}

\end{Theorem}

\begin{proof}
    Using Theorem \ref{curvatura} and the naturality of the exterior derivative, we obtain
   $$F\equiv \sigma ^{\ast}\Omega =\sigma ^{\ast}(d\omega + \omega \wedge \omega)=d(\sigma ^{\ast}\omega)+(\sigma ^{\ast}\omega)\wedge (\sigma ^{\ast}\omega)=dA+A\wedge A=dA+\frac{1}{2}[A\wedge A].$$
\end{proof}

\noindent The preceding expressions are typically written as elements of the Lie algebra of the Lie group $G$. Alternatively, if we express the connection and curvature in terms of generators, where we make the change
\begin{equation*}
    A\longrightarrow-iA,\;\;\; F\longrightarrow -iF,
\end{equation*}

\noindent then eq.\eqref{comp local curvature} is rewritten as
\begin{equation}
    F_{\mu \nu}=\partial_{\mu} A_{\nu}-\partial_{\nu} A_{\mu}-i[A_{\mu},A_{\nu}].
\end{equation}

\noindent After a little algebra it can be demonstrated that:
\begin{Lemma}
    Under gauge transformations, the local curvature transforms as
    \begin{equation}
        F\longrightarrow (Ad_{g^{-1}})_{\ast}F\equiv g^{-1}Fg,
    \end{equation}
    where the arrow notation indicates the change of gauge from a $\sigma _{1}$ to a $\sigma _{2}$.
\end{Lemma}

\begin{Theorem}\label{teoderivadacovarianteequivariante}

        Let $A\in \Omega ^{1}(M,\mathfrak{g})$ be a local connection and $B\in \Omega(M,\mathfrak{g})$ be a $\mathfrak{g}$-valued $G$-equivariant form, meaning that under gauge transformations, $B$ transforms as
        \begin{equation*}
            B\longrightarrow g^{-1}Bg.
        \end{equation*}
        Then, the covariant derivative of $B$ is given by
        \begin{equation}
        \label{derivadacovarianteequivariante}
            d_{A}B=dB+[A\wedge B].
        \end{equation}
        
    \end{Theorem}

\begin{Theorem}[The reduction theorem ~\cite{coquereaux1988riemannian}]
    There is a one-to-one correspondence between G-invariant metrics on P and triples $(g_{\mu \nu}, \omega, k_{ij})$, where $g_{\mu \nu}$ is a metric on $M$, $\omega$ is a 1-form connection, and $k_{ij}(x)$ is a choice of a $G$-invariant metric on each fiber $G_{x}$.
    \label{metricaginvariante}
\end{Theorem}

\subsection{Physical applications}

\noindent 
A well-known application of principal $G$-bundles is Kaluza-Klein theory, which unifies gravity and electrodynamics. The central element of Kaluza-Klein theory is a metric $\gamma_{AB}$ on $M\times S^{1}$, where $A,B$ are indices that runs from $0$ to $4$,  
$M$ is the $4$-dimensional spacetime and $S^{1}$ is the circle, along with the so-called cylindrical condition

\begin{equation}
\frac{\partial \gamma_{AB}}{\partial x^{4}}=0,
\end{equation}
\noindent where $x^{4}=y=\theta$ is the coordinate of the circle $S^{1}$. The gauge ambiguity can be eliminated by considering a $U(1)$-invariant metric $\gamma$ on a principal $U(1)$-bundle $\pi :P\rightarrow M$, which is locally isomorphic to $M\times S^{1}$.


\noindent The action integral of the theory is written as a 5-dimensional Einstein-Hilbert Lagrangian,
\begin{equation}
    S=\frac{1}{16 \pi \hat{G}}\int _{P} dx^{4}dy\sqrt{-\gamma}R^{(5)},
\end{equation}
where $\hat{G}$ is a 5-dimensional gravitational constant,  $\gamma$ is the determinant of $\gamma _{AB}$, and $R^{(5)}=R_{BAD}^{A}\gamma ^{BD}$ is the five dimensional scalar curvature of the metric $\gamma _{AB}$. The action can be expressed as
\begin{equation}
    S=\int _{M} dx^{4}\sqrt{-g}\left(\frac{1}{16 \pi {G}}R+\frac{1}{4}\phi ^{2}F_{\mu \nu}F^{\mu \nu}+\frac{1}{24\pi G}\frac{\partial ^{\mu}\phi \partial _{\mu}\phi}{\phi ^{2}}\right),
\end{equation}
where $G=\hat{G}/\int dy$ is
the gravitational constant and $F_{\mu\nu}$ is the field strength of the local $U(1)$ connection $A_\mu$ and $\phi$ is part of the internal metric (see~\cite{Cho:1975sw,overduin1997kaluza} for details). 

\noindent Another application of gauge theory is in theories of modified gravity that introduce additional fields. For example, one recently developed theory is $BF$-coupled gravity~\cite{alexander2020topological}, for which the action is given by
\begin{equation}\label{eq:seq}
	S=\int_M \bigl \langle B\wedge F\bigr \rangle+ \left[\frac{1}{2\kappa}R(\hat{g})-\frac{\bar{\Lambda}}{\kappa}+\mathscr{L}_{M}(\hat{g})\right] \;\bigl \langle B\wedge B\bigr \rangle=S_{BF}+S_{GR},
\end{equation}
where $\kappa=8\pi G c^{-4}$, and gravity is coupled to the $BF$ fields through 
\begin{equation}
\label{bwedgeb}
    \;\bigl \langle B\wedge B\bigr \rangle \equiv \frac{1}{4}d^{4}x\epsilon^{\mu\nu\rho\sigma}B_{\mu \nu}^{a}B_{\rho \sigma \;a}=:d^{4}x\;\sqrt{|\hat{g}|}\neq 0.
\end{equation} 
In the previous expressions the internal metric $\langle ,\rangle$ is used to lower internal indices, meaning that $B_{\rho \sigma \;a}=B_{\rho \sigma}^{b}\langle \tau _{b}, \tau _{a}\rangle$ where $\{\tau _{a}\}_{a=1,...,n}$ constitutes a basis for the Lie algebra of the Lie group $G$ of dimension $n$. The composite metric $\hat{g}_{\mu \nu}$ is defined as
\begin{equation}
    \hat{g}_{\mu \nu}=\left(\frac{\hat{g}}{g}\right)^{\frac{1}{4}}\;g_{\mu \nu},
\end{equation}
where $g_{\mu \nu}$ is an arbitrary metric, considered as one of the fundamental fields in the action and $g$ its determinant. Futhermore, the components of the inverse  metric are given by
\begin{equation}\label{eq:gghat}
    \hat{g}^{\mu \nu}=\left(\frac{\hat{g}}{g}\right)^{-\frac{1}{4}}\;g^{\mu \nu}.
\end{equation}
The transformations that leave $\hat{g}_{\mu \nu}$ invariant and consequently all terms in the action
~\eqref{eq:seq} are Weyl transformations, defined as 
\begin{equation}
    g_{\mu \nu}\longrightarrow \Omega^{2}(x)g_{\mu \nu},
\end{equation}
and transverse diffeomorphisms $\mathscr{L}_{\xi} \langle B\wedge B \rangle=0$, where $\mathscr{L}_{\xi}$ is the Lie derivative generated by a transverse vector field $\xi^{a}$. 

\noindent The metric $\hat{g}_{\mu \nu}$ is a non-trivial function of the fields $B$ and $g_{\mu \nu}$. The variation of the action with respect to these two fields can be evaluated using the chain rule, and is given in terms of the variation with respect to $\delta \hat{g}^{\mu \nu}$,
\begin{equation}
    \begin{split}
        \delta S_{GR} & \equiv \delta \int_M d^{4}x\sqrt{\hat{g}}\; \left[\frac{1}{2\kappa}R(\hat{g})-\frac{\bar{\Lambda}}{\kappa}+\mathscr{L}_{M}\right]\\[0.2cm]
        & = \int_M d^{4}x\; \left[\frac{1}{2\kappa}\frac{\delta(\sqrt{\hat{g}}R)}{\delta \hat{g}^{\mu \nu}}-\frac{\bar{\Lambda}}{\kappa}\frac{\delta\sqrt{\hat{g}}}{\delta \hat{g}^{\mu \nu}}+\frac{\delta(\sqrt{\hat{g}}\mathscr{L}_{M})}{\delta \hat{g}^{\mu \nu}}\right] \delta \hat{g}^{\mu \nu}\\[0.2cm]
        & = \int_M d^{4}x\frac{\sqrt{\hat{g}}}{2\kappa}\; \left[\left(\frac{\delta R}{\delta \hat{g}^{\mu \nu}}+\frac{(R-2\bar{\Lambda})}{\sqrt{\hat{g}}}\frac{\delta\sqrt{\hat{g}}}{\delta \hat{g}^{\mu \nu}}\right)+\kappa \frac{2}{\sqrt{\hat{g}}}\frac{\delta(\sqrt{\hat{g}}\mathscr{L}_{M})}{\delta \hat{g}^{\mu \nu}}\right] \delta \hat{g}^{\mu \nu}\\[0.2cm]
        & = \int_M d^{4}x\frac{\sqrt{\hat{g}}}{2\kappa}\; \left[R_{\mu \nu}(\hat{g})-\frac{1}{2}R\hat{g}_{\mu \nu}+\bar{\Lambda}\hat{g}_{\mu \nu} -\kappa T_{\mu \nu}\right]\delta \hat{g}^{\mu \nu},
    \end{split}
\end{equation} 
where we have employed the functional derivative identities for $R$ and $\hat{g}$,
\begin{equation*}
    \frac{\delta R(\hat{g})}{\delta \hat{g}^{\mu \nu}}=R_{\mu \nu}(\hat{g}),\;\;\;\;\; \frac{\delta\sqrt{\hat{g}}}{\delta \hat{g}^{\mu \nu}}=-\frac{1}{2}\sqrt{\hat{g}}\;\hat{g}_{\mu \nu},
\end{equation*}
 where we have omitted the boundary terms as they do not contribute to the variation and applied the definition of the energy-momentum tensor,
\begin{equation}
    T_{\mu \nu}(\hat{g})=\frac{-2}{\sqrt{\hat{g}}}\frac{\delta(\sqrt{\hat{g}}\mathscr{L}_{M})}{\delta \hat{g}^{\mu \nu}}=-2\frac{\delta\mathscr{L}_{M}}{\delta \hat{g}^{\mu \nu}}+\hat{g}_{\mu \nu}\mathscr{L}_{M}.
\end{equation}
Using the definition of the trace of the energy-momentum tensor $T=T_{\mu \nu}\hat{g}^{\mu \nu}$, we can finally evaluate the variation with respect to $B$,
{\allowdisplaybreaks
\begin{align}
\label{variacionbbf}
        0 & = \frac{\delta (S_{BF}+S_{GR})}{\delta B_{\mu ^{\prime} \nu ^{\prime}} ^{a}}\delta B_{\mu ^{\prime} \nu ^{\prime}} ^{a}  \nonumber\\[0.2cm]
        & =\int_{M}d^{4}x\;\left[ \epsilon^{\mu ^{\prime}\nu ^{\prime}\rho \sigma}\;\frac{1}{4}\;F_{\rho\sigma\;a}+\frac{\sqrt{\hat{g}}}{2\kappa}\; \left[G_{\mu \nu}+\bar{\Lambda}\hat{g}_{\mu \nu} -\kappa T_{\mu \nu}\right] \frac{\delta\hat{g}^{\mu \nu}}{\delta B_{\mu ^{\prime}\nu ^{\prime}} ^{a}}\right]\delta B_{\mu ^{\prime} \nu ^{\prime}} ^{a}\nonumber\\[0.2cm]
        & =\int_{M}d^{4}x\;\left[ \epsilon^{\mu ^{\prime}\nu ^{\prime}\rho\sigma}\;\frac{1}{4}\;F_{\rho\sigma\;a}-\frac{\sqrt{\hat{g}}}{8\kappa}\; \left[G_{\mu \nu}+\bar{\Lambda}\hat{g}_{\mu \nu} -\kappa T_{\mu \nu}\right]\frac{g^{1/4}}{\hat{g} ^{5/4}} g^{\mu \nu} \frac{\delta \hat{g}}{\delta B_{\mu ^{\prime}\nu ^{\prime}} ^{a}}\right]\delta B_{\mu ^{\prime} \nu ^{\prime}} ^{a}\nonumber\\[0.2cm]
        & =\int_{M}d^{4}x\;\left[ \;\frac{1}{4}\;F_{\rho\sigma\;a}+\frac{1}{8\kappa}\; \left[R(\hat{g})-4\bar{\Lambda} +\kappa T\right]B_{\rho \sigma\;a}\right]\epsilon^{\mu ^{\prime}\nu ^{\prime}\rho\sigma}\delta B_{\mu ^{\prime} \nu ^{\prime}} ^{a},
\end{align}
}
where
\begin{equation*}
    \frac{\delta \hat{g}}{\delta B_{\mu ^{\prime}\nu ^{\prime}} ^{a}}=\sqrt{\hat{g}}\;\epsilon^{\mu ^{\prime}\nu ^{\prime}\rho\sigma}B_{\rho \sigma\;a},
\end{equation*}
follows from eq.~\eqref{bwedgeb} and $G_{\mu\nu}$ is
the Einstein tensor of the metric $\hat g$. Using eq.~\eqref{variacionbbf}, 
the identity $\epsilon^{\mu \nu \rho \sigma}\epsilon_{\mu \nu \alpha \beta}=2(\delta _{\alpha}^{\rho}\delta _{\beta}^{\sigma}-\delta _{\beta}^{\rho}\delta _{\alpha}^{\sigma})$ and the antisymmetry properties of the components of $F$ and $B$, we obtain the field equation
\begin{equation}
\label{primeraecuacionbfsecuestered}
    F+\frac{1}{2\kappa}\left[R(\hat{g})+\kappa T-4\bar{\Lambda}\right]B=0.
\end{equation}
\noindent Furthermore, variation with respect to $A$ leads to
\begin{equation}
\label{segundaecuacionbfsecuestered}
    d_{A}B=0,
\end{equation}
since $S_{GR}$ does not depend on the connection $A$.
 Finally, variation with respect to $g^{\mu \nu}$ leads to the trace-free Einstein equation,
{\allowdisplaybreaks\begin{align}
\label{desarrollovariaciónrg}
        0 & =\frac{\delta S_{GR}}{\delta g^{\mu ^{\prime} \nu ^{\prime}}}\delta g^{\mu ^{\prime} \nu ^{\prime}}\nonumber\\[0.2cm]
        & =\int_M d^{4}x\frac{\sqrt{\hat{g}}}{2\kappa}\; \left[R_{\mu \nu}(\hat{g})-\frac{1}{2}R\hat{g}_{\mu \nu}+\bar{\Lambda}\hat{g}_{\mu \nu} -\kappa T_{\mu \nu}\right] \frac{\delta \hat{g}^{\mu \nu}}{\delta g^{\mu ^{\prime} \nu ^{\prime}}}\delta g^{\mu ^{\prime} \nu ^{\prime}}\nonumber\\[0.2cm]
        & =\int_M d^{4}x\frac{\sqrt{\hat{g}}}{2\kappa}\; \left[G_{\mu \nu}+\bar{\Lambda}\hat{g}_{\mu \nu} -\kappa T_{\mu \nu}\right] \frac{g^{1/4}}{\hat{g}^{1/4}}\left(\delta _{\mu ^{\prime}}^{\mu}\delta _{\nu ^{\prime}}^{\nu}-\frac{1}{4}g^{\mu \nu}g_{\mu ^{\prime} \nu ^{\prime}}\right)\delta g^{\mu ^{\prime} \nu ^{\prime}}\nonumber\\[0.2cm]
        & =\int_M d^{4}x\frac{\sqrt{\hat{g}}}{2\kappa} (\hat{g}/g)^{-1/4}\; \left[R_{\mu ^{\prime} \nu ^{\prime}}-\frac{1}{4}R\hat{g}_{\mu ^{\prime} \nu ^{\prime}} -\kappa \left(T_{\mu ^{\prime}\nu ^{\prime}}-\frac{1}{4}T\hat{g}_{\mu ^{\prime}\nu ^{\prime}}\right)\right]\delta g^{\mu ^{\prime} \nu ^{\prime}}.
\end{align}}
\noindent Thus, from the previous equation, we deduce that the gravitational field equation is
\begin{equation}
\label{tereceraecuacionbfsecuestered}
    R_{\mu \nu }(\hat{g})-\frac{1}{4}R(\hat{g})\hat{g}_{\mu \nu }=\kappa \left(T_{\mu \nu}-\frac{1}{4}T\hat{g}_{\mu \nu}\right).
\end{equation}
Equations~\eqref{primeraecuacionbfsecuestered}, \eqref{segundaecuacionbfsecuestered} and \eqref{tereceraecuacionbfsecuestered} are the field equations of the theory. From the Bianchi identity, eqs.~\eqref{primeraecuacionbfsecuestered}, \eqref{segundaecuacionbfsecuestered} and the condition $\hat{g}\neq 0$ it follows that
\begin{equation}
    d(R+\kappa T)=0.
\end{equation}
Thus we can write $R+\kappa T=4 \Lambda$, where $\Lambda$ is an integration constant that plays the role of the cosmological constant. This interpretation arises when noticing that eq.~\eqref{tereceraecuacionbfsecuestered} can be rewritten as
\begin{equation}
    R_{\mu \nu}(\hat{g})-\frac{1}{2}\hat{g}_{\mu \nu}R(\hat{g})+\Lambda \hat{g}_{\mu \nu}=\kappa T_{\mu \nu}.\label{eqeinstein}
\end{equation}
In GR the  cosmological constant is a fixed coupling constant and thus has a specific value (up to renormalizations). However, $\Lambda$ in eq.~\eqref{eqeinstein} is an integration constant that remains undetermined until a solution is chosen, and it is not affected by corrections, which instead modify $\bar{\Lambda}$. The field equations~\eqref{eqeinstein} have been proposed in the literature as \textit{unimodular gravity}. 
It can be argued that if one dismisses the naturalness problem of the cosmological constant, then GR is classically equivalent to UG. However, an important difference between both theories is that 
Birkhoff's theorem is no longer valid in UG~\cite{Bonder:2022kdw,Alvarez2023}, instead, vacuum, spherically symmetric solutions can be Schwarzschild, Schwarzschild-de Sitter or Schwarzschild-anti-de Sitter spacetimes. This allows for vacuum expanding cosmological solutions. The same conclusion can be reached in BF coupled to gravity, since eq.~\eqref{tereceraecuacionbfsecuestered} allows for constant curvature spacetimes in vacuum. This highlights the fact that Einstein equations and equations~\eqref{eqeinstein} are fundamentally different. Nevertheless, UG still describes a cosmological model where the accelerated expansion is driven by the constant energy density $\Lambda$ and, perhaps, by an additional fluid that is introduced \textit{ad-hoc} in the energy-momentum tensor. On the other hand, in BF coupled to gravity,  $\Lambda$ in eq.~\eqref{eqeinstein} can be promoted to a dynamical field by including in the action a kinetic term for the gauge field. This gives a theoretical explanation for the origin of dynamical dark energy. In~\cite{alexander2022black}, it is shown that action~\eqref{eq:seq} supplemented with a kinetic term for a $SU(2)$ gauge field and without a matter Lagrangian leads to a de Sitter cosmological model at late times. It is also interesting to note that the model with a $U(1)$ gauge field leads to astrophysical solutions that, in some limits, resemble Reissner-Nordstr\"om black holes (solutions to the Einstein-Maxwell theory) but introduce the notion of a fundamental unit charge. The previous discussion motivates the search for a way to couple different groups to gravity, since each group may have its own relevant consequences for different gravitational phenomena, including dark energy. In addition, it is natural to ask whether one can couple the gauge fields that are relevant for the standard model of particle physics. It was recently shown that  a classical action describing the standard model coupled to Einstein gravity can be obtained from a constrained BF theory based on a 3-group~\cite{Radenkovic:2019qme}. However, as discussed above, even from a classical perspective it is justified to look for alternatives to Einstein gravity. Here, we develop the fundamentals for obtaining a modified gravity theory that is coupled to several gauge fields, extending the approach presented in~\cite{alexander2020topological} to higher order BF theories.

\section{Categorical approach}\label{sec:cat}
In this section, we introduce the notions from category theory essential for categorizing the geometry coming from connections. We interpret connections as functors and introduce the first notions from higher categories to define 2-connections as 2-functors and motivate higher gauge theory~\cite{baez2005higher}. 

\subsection{The connection functor}

We aim to show that a connection over a principal $G$-bundle can be translated into the language of categories as a functor between two special categories, namely the path groupoid  and the induced category of $G$ (for an in-depth analysis in category theory see~\cite{mac2013categories})

\begin{Definition}
            A category $\mathscr{C}$ consists in the following data:
            \begin{enumerate}
	               \item A collection of objects $Obj(\mathscr{C})$.
	               \item A collection of arrows $Mor(\mathscr{C})$ between objects. An arrow $f$ is represented by
                          \begin{equation*}
                    \begin{tikzcd}[row sep=huge]
                        A \arrow [r, "f"] &
                        B,
                    \end{tikzcd}
                \end{equation*}
                where $A,B\in Obj(\mathscr{C})$
               
                \item An operation of composition of arrows. Given two arrows
    $  \begin{tikzcd}[row sep=huge]
                        A \arrow [r, "f"] &
                        B,
                    \end{tikzcd}
                    \begin{tikzcd}[row sep=huge]
                    B \arrow [r, "g"] &
                    C,
                    \end{tikzcd}$
                there is an arrow in $Mor(\mathscr{C})$
                \begin{equation*}
                    \begin{tikzcd}[row sep=huge]
                        A \arrow [r, "g\circ f"] &
                        C,
                    \end{tikzcd}
                \end{equation*}
                called the composite of $f$ and $g$.
                \end{enumerate}

           \noindent Satisfying the following conditions:
              \begin{enumerate}
        \renewcommand{\theenumi}{\alph{enumi}}
            \item (Identity) For each object $A$ there exists an arrow
                \begin{tikzcd}[row sep=huge]
                    A \arrow [r, "1_{A}"] &
                    A
                \end{tikzcd}
                such that for any \begin{tikzcd}[row sep=huge]
                A \ar [r, "f"] &
                B,
                \end{tikzcd} 
                \begin{equation}
                f\circ 1_{A}=f=1_{B}\circ f.
            \end{equation}
                         $1_A$ is called the identity arrow in $A$.
            
                  \item (Associativity) For any three arrows $\begin{tikzcd}[row sep=huge]
                A \arrow [r, "f"] &
                B,
                \end{tikzcd}
                \;
                \begin{tikzcd}[row sep=huge]
                B \arrow [r, "g"] &
                C,
                \end{tikzcd}
                \begin{tikzcd}[row sep=huge]
                C \arrow [r, "h"] &
                D,
                \end{tikzcd}$
            \begin{equation}
                h\circ (g\circ f)=(h\circ g)\circ f.
            \end{equation}
            
             \end{enumerate}
        \end{Definition}
        
    \begin{Example}[Category induced by a group] Let $G$ be a group, we define the category $\mathscr{G}$ by the data,
 \begin{align*}
                Obj(\mathscr{G})&=\{\ast\}, \nonumber \\
                Mor(\mathscr{G})&=G,
            \end{align*}
        where $\{\ast\}$ is a set with a unique element. An arrow corresponding with an element $g\in G$ is represented by
        \begin{equation*}
            \begin{tikzcd}[row sep=huge]
            \ast \arrow [r, "g"] &
            \ast.
            \end{tikzcd}
        \end{equation*}
        Given two arrows $g,h \in G$, their composition is defined as the product $gh\in G$. The identity arrow $1_{\ast} = e$ is the identity element of $G$.
    \end{Example}

    \begin{Definition}
            A groupoid is a category in which every arrow is invertible, that is, for
            $\begin{tikzcd}
                A \ar [r, "f"] &
                B
            \end{tikzcd}$ there exists an arrow
            $\begin{tikzcd}
                B \ar [r, "f^{-1}"] &
                A
            \end{tikzcd}$
            such that 
          $  f\circ f^{-1} =1_{B},\; 
                 f^{-1}\circ f  =1_{A}.$
                   \end{Definition}
\noindent Notice that the category induced by a group is also an example of a groupoid .

    \begin{Definition}
        A (covariant) functor
      $
            F:\mathscr{C}\rightarrow \mathscr{D}
       $
        between the categories $\mathscr{C}$ and $\mathscr{D}$, consists in two maps (one at level of objects and one at level of arrows)
        \begin{align*}
            F&:Obj(\mathscr{C})\rightarrow Obj(\mathscr{D}), \\
            F&:Mor(\mathscr{C})\rightarrow Mor(\mathscr{D}),
        \end{align*}
        satisfying the following conditions:
        \begin{enumerate}
        \renewcommand{\theenumi}{\alph{enumi}}
            \item (Domain/codomain) If
            $\begin{tikzcd}[row sep=huge]
                A \arrow [r, "f"] &
                B
            \end{tikzcd}$ in $Mor(\mathscr{C})$, then $ \begin{tikzcd}[row sep=huge]
                F(A) \arrow [r, "F(f)"] &
                F(B)
            \end{tikzcd}
            $ in $Mor(\mathscr{D})$.
            \item ($F$ preserves composition)
            $
                F(g\circ f)=F(g)\circ F(f).
            $  
            \item ($F$ preserves identities)
            $    F(1_{A})=1_{F(A)}.
            $ 
        \end{enumerate}
    \end{Definition}
  
    \begin{Definition}
        Let $\gamma: [0,1] \rightarrow M$ be a smooth path on $M$. $\gamma$ is called a ``lazy path'' if it remains constant in the neighborhoods of $t=0$ and $t=1$. Such a path is denoted as $\gamma: x \rightarrow y$, where $x = \gamma(0)$ is the initial point and $y = \gamma(1)$ is the final point.
    \end{Definition}

    \begin{Definition}
        A \emph{thin homotopy} between two lazy paths $\gamma$, $\delta: x \rightarrow y$ is a differentiable function $H: [0, 1] \times [0, 1] \rightarrow M$ that satisfies the following conditions:

         \begin{itemize}
            \item $H(0,t)=\gamma (t)$,\;\; $H(1,t)=\delta (t)$.
            \item $H(s,0)=x,\;\; H(s,1)=y$ for all $s\in [0,1]$.
            \item Its derivative has a rank less than 2 everywhere, meaning that
            \begin{equation*}
                H_{\ast}:T_{(s,t)}([0,1]\times [0,1])\rightarrow T_{H(s,t)}M
            \end{equation*} has a rank less than two for any $(s, t) \in [0, 1] \times [0, 1]$.
        \end{itemize}
    \end{Definition}

    \begin{Definition}
        Two lazy paths, $\gamma$ and $\delta$, are \emph{thin homotopic} if there exists a thin homotopy between them. The collection of lazy paths that are thin homotopic to $\gamma$ is called the thin homotopy class of $\gamma$ and is denoted by $[\gamma]$.
    \end{Definition}

    \begin{Definition}
        The composition of two thin homotopy classes is defined as $
            [\gamma][\delta]=[\gamma \delta]
        $,
        where
        \begin{equation*}
            \gamma \delta(t)=
            \begin{cases}
                            \gamma(2t) &         \text{if } 0\leq t \leq \frac{1}{2},\\
                            \delta(2t-1), &         \text{if } \frac{1}{2}\leq t \leq 1.
                    \end{cases}
        \end{equation*}
    \end{Definition}

    \begin{Definition}
        The path groupoid $\mathcal{P}_{1}(M)$ of $M$ is the category defined as:
        \begin{itemize}
            \item Objects are points in $M$.
            \item Arrows are thin homotopy classes of lazy paths in $M$, $\begin{tikzcd}
                x \ar [r, "{[\gamma ]}"] &
                y
            \end{tikzcd}$
            where $x=\gamma (0)$ and $y=\gamma(1)$.
            \item Composition is the composition of equivalence classes $
               [\gamma][\delta]=[\gamma \delta].$
           \item For any $x\in M$, the identity $1_{x}$ is the thin homotopy class of the constant path at $x$.
           \item The inverse of an arrow $[\delta]$ is $[\delta ^{-1}]$ where
          $\delta ^{-1}(t)=\delta (1-t)$.
        \end{itemize}
   \end{Definition}
    
    \begin{Remark}
        Given a principal $G$-bundle $\pi :P\rightarrow M$ equipped with a 1-form connection $\omega$ and a cover of trivializations (given by sections) $\{(U_{i},\sigma _{i})\}_{i\in I}$, where $\pi \circ \sigma _{i}=id_{U_{i}}$ for all $i\in I$, we have functors from the path groupoid  to the category defined by $G$,
    \begin{equation}
       hol_{i}:\mathcal{P}_1(U_{i})\rightarrow G,
   \end{equation}
  defined at objects and arrows,
    \begin{equation*}
        \begin{split}
            hol_{i}:U_{i} & \rightarrow \{\ast\}\\
            x & \mapsto \ast.
        \end{split}
   \end{equation*}
    \begin{align*}
            hol_{i}:Mor  (\mathcal{P}  (U_{i}))  &\rightarrow G\\ \nonumber
            [\gamma] & \mapsto \mathcal{P}\textup{exp}(\int _{\gamma}-A_{i}),
    \end{align*}
    where $A_{i}=\sigma _{i}^{\ast} \omega$. In this sense connections can be seen as functors (see~\cite{baez2011invitation, schreiber2011smooth} for details).
    \end{Remark}

\subsection{Strict 2-categories and strict 2-functors}

From the above construction a connection can be interpreted as a functor from the path groupoid  to the group. In category theory a very useful technique to generalize the notion of a category is by adding information to get a higher category, that is, a category in the usual sense together with additional data, namely 2-arrows (``arrows between the arrows"). In this section we aim to introduce the notions of strict 2-category and strict 2-functor. 

\begin{Definition}
    A strict 2-category $\mathscr{C}$ consists in the following data:
    \begin{enumerate}
        \item A collection of objects $Obj(\mathscr{C})$.
        \item A collection of 1-arrows $Mor(\mathscr{C})$ between objects. A 1-arrow $f$ is represented by
                          \begin{equation*}
                    \begin{tikzcd}[row sep=huge]
                        A \arrow [r, "f"] &
                        B.
                    \end{tikzcd}
                \end{equation*}
                where $A,B\in Obj(\mathscr{C})$.
               
        \item A composition of 1-arrows, which together with the objects form a category.
        \item A collection of 2-arrows 2-$Mor(\mathscr{C})$ between 1-arrows. A 2-arrow is represented by
       
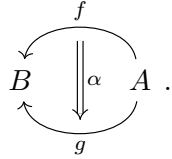
\begin{equation*} \begin{tikzcd}
                    B &
                    A \arrow[bend right=80,swap,"f"]{l}[name=RUU, below]{}
                    \arrow[bend left=80,"g"]{l}[name=RDD, above]{}
                    \arrow[Rightarrow,to path=(RUU) -- (RDD)\tikztonodes]{r}{\alpha}
                \end{tikzcd}.
                \end{equation*}
                where $A,B\in Obj(\mathscr{C})$ and $f,g\in Mor(\mathscr{C})$.
                \item 2-arrows can be composed in two different ways:
            \begin{itemize}
                \item Vertically,
                    
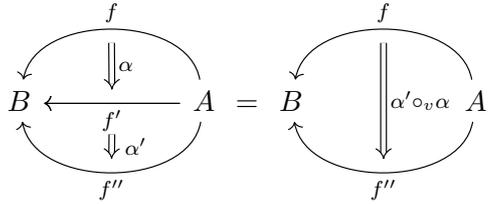
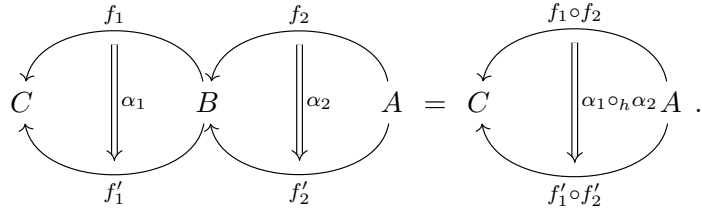
\begin{equation*}
                    \begin{tikzcd}[column sep=1.8cm]
                        B &
                        A \arrow[bend right=80,swap,"f"]{l}[name=1, below]{}
                        \arrow[bend left=80,"f^{\prime \prime}"]{l}[name=3, above]{}
                        \arrow["f^{\prime}"]{l}[name=12, above]{}
                        \arrow[swap]{l}[name=23,below=6.5]{}
                        \arrow[Rightarrow,to path=(1) -- (12)\tikztonodes]{r}{\alpha}
                        \arrow[Rightarrow,to path=(23) -- (3)\tikztonodes]{r}{\;\alpha ^{\prime}}
                    \end{tikzcd}
                =
                    \begin{tikzcd}[column sep=1.8cm]
                        B &
                        A \arrow[bend right=80,swap,"f"]{l}[name=1, below]{}
                        \arrow[bend left=80,"f^{\prime \prime}"]{l}[name=2, above]{}
                        \arrow[Rightarrow,to path=(1) -- (2)\tikztonodes]{r}{\alpha^{\prime} \circ _{v} \alpha}
                    \end{tikzcd}.
                \end{equation*}
                \item Horizontally,
                \begin{equation*}
                    \begin{tikzcd}[column sep=1.8cm]
                        C &
                        B \arrow[bend right=80,swap,"f_{1}"]{l}[name=1l, below]{}
                        \arrow[bend left=80,"f_{1}^{\prime}"]{l}[name=2l, above]{}
                        \arrow[Rightarrow,to path=(1l) -- (2l)\tikztonodes]{r}{\alpha_{1}} &
                        A \arrow[bend right=80,swap,"f_{2}"]{l}[name=1r, below]{}
                        \arrow[bend left=80,"f_{2}^{\prime}"]{l}[name=2r, above]{}
                        \arrow[Rightarrow,to path=(1r) -- (2r)\tikztonodes]{r}{\alpha_{2}}
                    \end{tikzcd}
                    =
                    \begin{tikzcd}[column sep=1.9cm]
                        C &
                        A\arrow[bend right=80,swap,"f_{1}\circ f_{2}"]{l}[name=1, below]{}
                        \arrow[bend left=80,"f_{1}^{\prime}\circ f_{2}^{\prime}"]{l}[name=2, above]{}
                        \arrow[Rightarrow,to path=(1) -- (2)\tikztonodes]{r}{\alpha_{1}\circ _{h} \alpha_{2}}
                    \end{tikzcd}.
                \end{equation*}
                \end{itemize}
    \end{enumerate}
    Satisfying the following properties:
    \begin{enumerate}
        \renewcommand{\theenumi}{\alph{enumi}}
        \item Vertical and horizontal compositions are associative.
            \item For each 1-arrow $\begin{tikzcd}[row sep=huge]
                A \arrow [r, "f"] &
                B
            \end{tikzcd}$ there exists a 2-arrow,
            
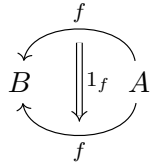
\begin{equation*}
                \begin{tikzcd}
                    B &
                    A \arrow[bend right=80,swap,"f"]{l}[name=1, below]{}
                    \arrow[bend left=80,"f"]{l}[name=2, above]{}
                    \arrow[Rightarrow,to path=(1) -- (2)\tikztonodes]{r}{1_{f}}
                \end{tikzcd}
            \end{equation*}
            acting as the identity in vertical composition. The 2-arrow 
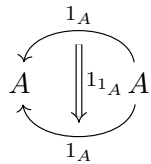
\begin{equation*}
                \begin{tikzcd}
                    A &
                    A \arrow[bend right=80,swap,"1_{A}"]{l}[name=1, below]{}
                    \arrow[bend left=80,"1_{A}"]{l}[name=2, above]{}
                    \arrow[Rightarrow,to path=(1) -- (2)\tikztonodes]{r}{1_{1_{A}}}
                \end{tikzcd}
            \end{equation*}
            serves as the identity for horizontal composition.

            \item Vertical and horizontal compositions follow the interchange law,
            \begin{equation}
                (\alpha_{1}^{\prime}\circ_{v}\alpha_{1})\circ _{h}(\alpha_{2}^{\prime}\circ_{v}\alpha_{2})=(\alpha_{1}^{\prime}\circ_{h}\alpha_{2}^{\prime})\circ _{v}(\alpha_{1}\circ_{h}\alpha_{2}),
            \end{equation}
            meaning there is no ambiguity in the composition of the following diagram
            \begin{equation*}
                    \begin{tikzcd}[column sep=1.8cm]
                        C &
                        B \arrow[bend right=80,swap,"f_{1}"]{l}[name=1l, below]{}
                        \arrow[bend left=80,"f_{1}^{\prime \prime}"]{l}[name=3l, above]{}
                        \arrow["\;f_{1}^{\prime}"]{l}[name=12l, above]{}
                        \arrow[swap]{l}[name=23l,below=6.5]{}
                        \arrow[Rightarrow,to path=(1l) -- (12l)\tikztonodes]{r}{\alpha_{1}}
                        \arrow[Rightarrow,to path=(23l) -- (3l)\tikztonodes]{r}{\;\alpha_{1} ^{\prime}} &
                        A \arrow[bend right=80,swap,"f_{2}"]{l}[name=1r, below]{}
                        \arrow[bend left=80,"f_{2}^{\prime \prime}"]{l}[name=3r, above]{}
                        \arrow["\;f_{2}^{\prime}"]{l}[name=12r, above]{}
                        \arrow[swap]{l}[name=23r,below=6.5]{}
                        \arrow[Rightarrow,to path=(1r) -- (12r)\tikztonodes]{r}{\alpha_{2}}
                        \arrow[Rightarrow,to path=(23r) -- (3r)\tikztonodes]{r}{\;\alpha_{2} ^{\prime}}
                    \end{tikzcd}.
            \end{equation*}
    \end{enumerate}
\end{Definition}
\begin{Definition}
    A strict 2-groupoid is a strict 2-category that satisfies the following conditions,
    \begin{itemize}
        \item Every 1-arrow  $\begin{tikzcd}
                A \ar [r, "f"] &
                B
            \end{tikzcd}$ has an inverse $\begin{tikzcd}
                B \ar [r, "f^{-1}"] &
                A
            \end{tikzcd}$,
            such that
            $$
                f\circ f^{-1}  =1_{B}, \;\; \; f^{-1}\circ f  =1_{A}.
          $$
            \item Every 2-arrow $\begin{tikzcd}
                f \ar [r,Rightarrow, "\alpha"] &
                g
            \end{tikzcd}$ has a vertical inverse $\begin{tikzcd}
                g \ar [r,Rightarrow, "\alpha_v^{-1}"] &
                f
            \end{tikzcd}$
            such that
           $$
               \alpha \circ _{v}\alpha _{v}^{-1}  =1_{g} ,\; \; \; \alpha _{v}^{-1}\circ _{v}\alpha =1_{f} .$$
            \item Every 2-arrow $\begin{tikzcd}
                f \ar [r,Rightarrow, "\alpha"] &
                g
            \end{tikzcd}$ where $f,g: A\rightarrow B$ has a horizontal inverse $\begin{tikzcd}
                f^{-1} \ar [r,Rightarrow, "\alpha _{h}^{-1}"] &
                g^{-1}
            \end{tikzcd}$ such that
            $$\alpha \circ _{h}\alpha _{h}^{-1} =1_{1_{B}}, \;\; \; \alpha _{h}^{-1}\circ _{v}\alpha  =1_{1_{A}} .$$
    \end{itemize}
\end{Definition}
\begin{Definition}
    A strict 2-group is a strict 2-groupoid with a unique object. 
\end{Definition}

\begin{Definition}
    A strict 2-functor $ F:\mathscr{C}\rightarrow \mathscr{D}$
        between the strict 2-categories $\mathscr{C}$ and $\mathscr{D}$, is a map at three levels:  objects, 1-arrows and 2-arrows,
        \begin{align*}
            F&:Obj(\mathscr{C})\rightarrow Obj(\mathscr{D}),
            \nonumber \\
            F&:Mor(\mathscr{C})\rightarrow Mor(\mathscr{D}),
            \nonumber \\       F&:2\textup{-}Mor(\mathscr{C})\rightarrow 2\textup{-}Mor(\mathscr{D}),
        \end{align*}
        such that:
     \begin{enumerate}
        \renewcommand{\theenumi}{\alph{enumi}}
            \item (Functor) The maps at level of objects and 1-arrows define a functor.
            \item (Domain/codomain) If $
            \begin{tikzcd}
                f \ar [r,Rightarrow, "\alpha"] &
                g
            \end{tikzcd}$ in 2-$Mor(\mathscr{C})$, then $\begin{tikzcd}
                F(f) \ar [r,Rightarrow, "F(\alpha)"] &
                F(g)
            \end{tikzcd}$ in 2-$Mor(\mathscr{D})$.
            \item ($F$ preserves horizontal and vertical compositions)
$F(\alpha\circ_{h_{\mathscr{C}}}\beta)=F(\alpha)\circ_{h_{\mathscr{D}}}F(\beta)$ and $ F(\alpha\circ_{v_{\mathscr{C}}}\beta)=F(\alpha)\circ_{v_{\mathscr{D}}}F(\beta).$
           \item ($F$ preserves identities for 2-arrows) $F(1_{f})=1_{F(f)}\;\;\;\;\; \textup{for any 1-arrow}
                \begin{tikzcd}[row sep=huge]
                A \arrow [r, "f"] &
                B
            \end{tikzcd}.$
        \end{enumerate}
\end{Definition}

\subsection{Crossed modules and strict 2-groups}

In this section we study the equivalence between crossed modules and  strict 2-groups, the approach of describing strict 2-groups through crossed modules will prove to be useful for the rest of this work.

\begin{Definition}
    A crossed module $(G,H,\partial ,\triangleright)$ consists in the following data:
    \begin{enumerate}
        \item Groups $G$ and $H$
        \item A group homomorphism $\partial: H \rightarrow G$.
        \item An action $\triangleright:G\curvearrowright H$ of $G$ on $H$.
    \end{enumerate}
    satisfying:
    \begin{enumerate}
        \renewcommand{\theenumi}{\alph{enumi}}
        \item The function $\phi_g: H \rightarrow H$ defined as $\phi_g(h)= g \triangleright h$ for any $h \in H$ belongs to the set of automorphisms of $H$.
        \item $\partial$ is $G$-equivariant:     $\partial(g\triangleright h)=g\partial (h)g^{-1}\;\; \forall g\in G, \forall h\in H$.
        \item The Peiffer identity holds: $
                \partial(h) \triangleright f=hfh^{-1}\;\; \forall h,f\in H.
             $
    \end{enumerate}
\end{Definition}
\noindent There is a natural pair of inverse equivalences between the category of strict 2-groups and the category
of crossed modules, in this sense these two categories are equivalent. We describe the effect of these functors on objects:

\begin{Remark}[Crossed module associated to a strict 2-group]
    Given a strict 2-group $\mathscr{G}$ there exists a crossed module $(G,H,\partial,\triangleright)$, where:
\begin{itemize}
    \item $G$ is the group of 1-arrows;
    \item $H$ is the group defined as the collection of the 2-arrows in $\mathscr{G}$ coming out of $1_{\ast}$,
    \begin{equation*}
                    \begin{tikzcd}[column sep=1.8cm]
                        \ast &
                        \ast \arrow[bend right=80,swap,"1_{\ast}"]{l}[name=1l, below]{}
                        \arrow[bend left=80,"\partial (h)"]{l}[name=2l, above]{}
                        \arrow[Rightarrow,to path=(1l) -- (2l)\tikztonodes]{r}{h} &
                        \ast \arrow[bend right=80,swap,"1_{\ast}"]{l}[name=1r, below]{}
                        \arrow[bend left=80,"\partial (h^{\prime})"]{l}[name=2r, above]{}
                        \arrow[Rightarrow,to path=(1r) -- (2r)\tikztonodes]{r}{h^{\prime}}
                    \end{tikzcd}
                   =
                    \begin{tikzcd}[column sep=2cm]
                        \ast &
                        \ast \arrow[bend right=80,swap,"1_{\ast}"]{l}[name=1, below]{}
                        \arrow[bend left=80,"\partial (hh^{\prime})"]{l}[name=2, above]{}
                        \arrow[Rightarrow,to path=(1) -- (2)\tikztonodes]{r}{\;h\circ _{h} h^{\prime}}
                    \end{tikzcd}\;;
                \end{equation*}
            \item The assignment
           $ \partial :H  \rightarrow G, \; \;\;
                h  \mapsto \partial (h),
            $  is a group homomorphism.
            \item  The action $\triangleright$ of $G$ on $H$ is defined as $            g\triangleright h=1_{g}\circ _{h}h\circ _{h}1_{g^{-1}},$ that is,
        \begin{equation*}
                    \begin{tikzcd}[column sep=1.8cm]
                        \ast &
                        \ast \arrow[bend right=80,swap,"g"]{l}[name=1ll, below]{}
                        \arrow[bend left=80,"g"]{l}[name=2ll, above]{}
                        \arrow[Rightarrow,to path=(1ll) -- (2ll)\tikztonodes]{r}{1_{g}} &
                        \ast \arrow[bend right=80,swap,"1_{\ast}"]{l}[name=1l, below]{}
                        \arrow[bend left=80,"\partial (h)"]{l}[name=2l, above]{}
                        \arrow[Rightarrow,to path=(1l) -- (2l)\tikztonodes]{r}{h} &
                        \ast \arrow[bend right=80,swap,"g^{-1}"]{l}[name=1r, below]{}
                        \arrow[bend left=80,"g^{-1}"]{l}[name=2r, above]{}
                        \arrow[Rightarrow,to path=(1r) -- (2r)\tikztonodes]{r}{1_{g^{-1}}}
                    \end{tikzcd}
                   =
                    \begin{tikzcd}[column sep=2cm]
                        \ast &
                        \ast \arrow[bend right=80,swap,"1_{\ast}"]{l}[name=1, below]{}
                        \arrow[bend left=80,"g\partial (h)g^{-1}"]{l}[name=2, above]{}
                        \arrow[Rightarrow,to path=(1) -- (2)\tikztonodes]{r}{\;g\;\triangleright h}
                    \end{tikzcd}.
                \end{equation*}
\end{itemize}
\end{Remark}

\begin{Remark}[strict 2-group associated to a crossed module] Given a crossed module $(G,H,\partial ,\triangleright)$ there exists a strict 2-group defined by the data:
\begin{itemize}
    \item There is a single object $\ast$.
    \item The 1-arrows are elements of the group $G$.
    \item The 2-arrows $\alpha :g \Rightarrow g^{\prime}$ are pairs $(g,h)\in G\times H$ with $g^{\prime}=\partial (h)g$.
    \item Vertical composition of $(g,h)$ and $(g^{\prime},h^{\prime})$, when composable, is given by
    \begin{equation*}
            (g,h)\circ _{v}(g^{\prime},h^{\prime})=(g^{\prime},hh^{\prime}).
    \end{equation*}
    \item  Horizontal composition of $(g,h)$ and $(g^{\prime},h^{\prime})$ is given by
    \begin{equation*}
            (g,h)\circ _{h}(g^{\prime},h^{\prime})=(gg^{\prime},h(g\triangleright h^{\prime})).
        \end{equation*}

\end{itemize}
\end{Remark}

\noindent An important concept arising from crossed modules is a differential crossed module, which plays a role in a 2-Lie group similar to the role a Lie algebra plays in a Lie group. In the following, we provide definitions for certain objects associated with a differential crossed module, illustrative examples are provided in Appendix \ref{app:modulocruzadodiferencial}. Importantly, it should be noted that the notation used here is non-standard and has been especially developed to avoid confusion that may arise when using standard notation~\cite{baez2003higher, martins2011lie}.

\begin{Definition}
    A differential crossed module $(\mathfrak{g},\mathfrak{h},\partial _{\ast}, \triangleright ^{\prime})$ over a crossed module $(G,H,\partial ,\triangleright)$ consists in
\begin{itemize}
    \item Lie algebras $\mathfrak{g} , \mathfrak{h}$ associated to the Lie groups $G,H$, respectively.
    \item A Lie algebra morphism $\partial _{\ast}: \mathfrak{h}\rightarrow \mathfrak{g}$ , induced by the push-forward (or differential) of the morphism of Lie groups $\partial :H\rightarrow G$ at $1_{H}$.
\item A left action of $\mathfrak{g}$ on $\mathfrak{h}$, $\triangleright ^{\prime}:\mathfrak{g} \rightarrow End(\mathfrak{h})$ induced by the push-forward of the linear function
\begin{equation*}
            \begin{tikzcd}
                \triangleright ^{\prime \prime}:G \arrow[r,"\triangleright"] & Aut(H) \arrow[r,"\ast"] & GL(\mathfrak{h}),
            \end{tikzcd}
        \end{equation*}
        defined as
        \begin{equation*}
        \triangleright ^{\prime \prime} (g)=(\phi _{g})_{\ast}\;\;\;\; \forall g\in G,
    \end{equation*}
    where 
    $\phi _{g}(h)=g\triangleright h$ for any $h\in H$ and $(\phi _{g})_{\ast}$ is its push-forward.
\end{itemize}
\end{Definition}

\subsection{2-connections}

We finally introduce the notion of 2-connections motivated by 2-functors, or more specifically, local 2-connections, which suffice for the purpose of this paper. A complete discussion can be found in ~\cite{schreiber2008connections}. In this section we define the 2-groupoid  of paths and similarly to the above construction we translate a 2-connection as a 2-functor from the 2-groupoid  of paths to a strict 2-group (or equivalently a crossed module). The last part of this section is to introduce the curvature, the fake curvature and the 3-form curvature arising from this approach.

\begin{Definition}
    A lazy surface $\begin{tikzcd}
                \gamma \ar [r,Rightarrow, "H"] &
                \delta 
    \end{tikzcd}$, between two lazy paths $\gamma, \delta:x\rightarrow y$ is a differentiable function $H:[0,1]^{2}\rightarrow M$ satisfying
    \begin{itemize}
        \item $H(0,t)=\gamma (t),\;\;\; H(1,t)=\delta (t)$;
        \item $H(s,t)$ is independent of $s$ near $s=0$ and $s=1$;
        \item $H(s,t)$ is constant near $t=0$ and constant near $t=1$.
    \end{itemize}
\end{Definition}
\begin{Definition}
    A thin homotopy between two lazy surfaces $H, H^{\prime}: \gamma \Rightarrow \delta$ is a differentiable function $K:[0,1]^{3}\rightarrow M$ such that,
    \begin{itemize}
            \item $K(0,s,t)=H(s,t),\;\;\;K(1,s,t)=H^{\prime}(s,t)$;
            \item $K(r,0,t)=\gamma (t), \;\;\;K(r,1,t)=\delta (t)$ for any $r\in [0,1]$;
            \item Its derivative has a rank less than $3$ everywhere.
        \end{itemize}
\end{Definition}

\begin{Definition}
    Two lazy surfaces $H, H^{\prime}: \gamma \Rightarrow \delta$ are \emph{thin homotopic} if there exists a thin homotopy between them. The collection of lazy surfaces that are thin homotopic to $H$ is called the thin homotopy class of $H$ and is denoted as $[H]$.
\end{Definition}

\begin{Definition}
    The 2-groupoid of paths on $M$, $\mathcal{P}_{2}(M)$, is the strict 2-category where:
    \begin{itemize}
        \item Objects are points in $M$.
        \item Arrows are thin homotopy classes of lazy paths in $M$.
        \item 2-arrows between thin homotopy classes of lazy paths, $[\gamma_{0}], [\gamma _{1}]:x\rightarrow y$, are thin homotopy classes of lazy surfaces,
        $\begin{tikzcd}
                {[\gamma_{0}]} \ar [r,Rightarrow, "{[H]}"] &
                {[\gamma_{1}]}.
                \end{tikzcd}$
                \item Horizontal composition is the usual composition of homotopies,
                \begin{equation*}
                    \begin{tikzcd}[column sep=1.8cm]
                        z &
                        y \arrow[bend right=70,swap,"{[\gamma _{1}]}"]{l}[name=1l, below]{}
                        \arrow[bend left=70,"{[\gamma _{1}^{\prime}]}"]{l}[name=2l, above]{}
                        \arrow[Rightarrow,to path=(1l) -- (2l)\tikztonodes]{r}{{[H_{1}]}} &
                        x \arrow[bend right=70,swap,"{[\gamma _{2}]}"]{l}[name=1r, below]{}
                        \arrow[bend left=70,"{[\gamma _{2}^{\prime}]}"]{l}[name=2r, above]{}
                        \arrow[Rightarrow,to path=(1r) -- (2r)\tikztonodes]{r}{{[H_{2}]}}
                    \end{tikzcd}
                    =
                    \begin{tikzcd}[column sep=2.45cm]
                        z &
                        \;x \arrow[bend right=80,swap,"{[\gamma _{1}\gamma _{2}]}"]{l}[name=1, below]{}
                        \arrow[bend left=80,"{[\gamma _{1}^{\prime}\gamma _{2}^{\prime}]}"]{l}[name=2, above]{}
                        \arrow[Rightarrow,to path=(1) -- (2)\tikztonodes]{r}{{[H_{1}\circ _{h}H_{2}]}}
                    \end{tikzcd},
                \end{equation*}
                where
                \begin{equation*}
                    H_{1} \circ _{h} H_{2}=
                    \begin{cases}
                            H_{2}(s,2t), &         \text{if } 0\leq t \leq \frac{1}{2}\\
                            H_{1}(s,2t-1), &         \text{if } \frac{1}{2}\leq t \leq 1
                    \end{cases}
                \end{equation*}
            \item Vertical composition is
            \begin{equation*}
                    \begin{tikzcd}[column sep=1.8cm]
                        y &
                        x \arrow[bend right=70,swap,"{[\gamma]}"]{l}[name=1, below]{}
                        \arrow[bend left=70,"{[\gamma ^{\prime \prime}]}"]{l}[name=3, above]{}
                        \arrow["{[\gamma ^{\prime}]}"]{l}[name=12, above]{}
                        \arrow[swap]{l}[name=23,below=6.5]{}
                        \arrow[Rightarrow,to path=(1) -- (12)\tikztonodes]{r}{{[H]}}
                        \arrow[Rightarrow,to path=(23) -- (3)\tikztonodes]{r}{{\;\;[H^{\prime}]}}
                    \end{tikzcd}
                =
                    \begin{tikzcd}[column sep=2.3cm]
                        y &
                        \;x \arrow[bend right=70,swap,"{[\gamma]}"]{l}[name=1, below]{}
                        \arrow[bend left=70,"{[\gamma^{\prime \prime}]}"]{l}[name=2, above]{}
                        \arrow[Rightarrow,to path=(1) -- (2)\tikztonodes]{r}{{[H^{\prime} \circ _{v} H]}}
                    \end{tikzcd},
                \end{equation*}
                where
                \begin{equation*}
                    H^{\prime} \circ _{v} H=
                    \begin{cases}
                            H(2s,t), &         \text{if } 0\leq s \leq \frac{1}{2}.\\
                            H^{\prime}(2s-1,t), &         \text{if } \frac{1}{2}\leq s \leq 1.
                    \end{cases}
                \end{equation*}
    \end{itemize}
\end{Definition}
\noindent Just as a local connection can be viewed as a functor, a local 2-connection can be seen as a strict 2-functor (see~\cite{baez2011invitation, schreiber2011smooth} for details)
\begin{equation}
    hol:\mathcal{P}_{2}(U)\rightarrow \mathcal{G}
\end{equation}
for some Lie strict 2-group $\mathcal{G}$. This 2-functor is equivalent to a tuple $(A,\beta)$ where $A$ is a $\mathfrak{g}$-valued 1-form (local connection) and $\beta$ is a $\mathfrak{h}$-valued 2-form. With this 2-connection, the corresponding curvature is
\begin{align}
        \mathcal{F}_{A,\beta}&=F_{A}-\partial _{\ast}\beta \equiv dA+A\wedge A-\partial _{\ast}\beta,\nonumber \\
        \mathcal{G}_{A,\beta}&=d\beta +A\wedge ^{\triangleright ^{\prime}}\beta,
    \end{align}
referred to as fake curvature and $3$-curvature, respectively.\\
\noindent In this generalization, gauge transformations are described by what is known as a natural pseudotransformation between the functors of holonomy:
\begin{Theorem}
    Let $hol ^{\prime}, hol:\mathcal{P}_{2}(U)\rightarrow \mathcal{G}$ be smooth $2$-functors with associated $1$-forms $A^{\prime},A\in \Omega^{1}(U,\mathfrak{g})$ and $2$-forms $\beta ^{\prime},\beta \in \Omega^{2}(U,\mathfrak{h})$ respectively. The smooth function $g:U\rightarrow G$ and the $1$-form $\eta \in\Omega^{1}(U,\mathfrak{h})$ extracted from a smooth pseudonatural transformation $\rho:hol^{\prime}\rightarrow hol$ satisfy the relations
    \begin{equation}
        A+\partial _{\ast}(\eta)=gA^{\prime}g^{-1}-(dg)g^{-1}
    \end{equation}
    \begin{equation}
        \beta +A\wedge ^{\triangleright ^{\prime}}\eta+d\eta+\eta \wedge\eta=g\triangleright ^{\prime \prime}\beta ^{\prime}
    \end{equation}
\end{Theorem}
\begin{proof}
    See  Schreiber and Waldorf~\cite{schreiber2011smooth}
\end{proof}
\noindent If we choose a smooth function $g:U\rightarrow G$ and a $1$-form $\eta \in\Omega^{1}(U,\mathfrak{h})$ that satisfies the relations of the previous theorem, then as shown in \cite{schreiber2011smooth}, a pseudonatural transformation $\rho:hol\rightarrow hol^{\prime}$ can be defined, with $g$ and $\eta$ as its extracted data. Thus we can define two type of transformations, one resembling a usual gauge transformation and another one that is proper of the generalization.

\begin{itemize}
    \item Thin gauge transformations:
A smooth function ``gauge transformation'' $g:M\rightarrow G$ with $\eta =0$ induces the transformations
        \begin{align}
            A & \longrightarrow g^{-1}Ag+g^{-1}dg,\nonumber \\
            \beta & \longrightarrow g^{-1}\triangleright ^{\prime \prime}\beta,
        \end{align}
        which in turn cause the curvature $F$, the fake curvature $\mathcal{F}$ and the $3$-form curvature $\mathcal{G}$ to transform as
        \begin{align}
        \label{tdelgadas}
             F&\longrightarrow g^{-1}Fg,\nonumber \\ \mathcal{F}&\longrightarrow g^{-1}\mathcal{F}g,\nonumber \\
             \mathcal{G}&\longrightarrow g^{-1}\triangleright ^{\prime \prime}\mathcal{G}.
        \end{align}
    \item Fat gauge transformations:
Given a $\mathfrak{h}$-valued 1-form $\eta$, the transformations with $g:U\rightarrow G$ trivial,
are
        \begin{align}
            A&\longrightarrow A+\partial _{\ast}(\eta),\nonumber \\
            \beta& \longrightarrow \beta +d\eta+A\wedge ^{\triangleright ^{\prime}}\eta +\eta \wedge \eta,
        \end{align}
        under which the curvature $F$, the fake curvature $\mathcal{F}$ and the $3$-form curvature $\mathcal{G}$ transform as
        \begin{align}
        \label{tgruesas}
            F & \longrightarrow F+\partial _{\ast}(d\eta+A\wedge ^{\triangleright ^{\prime}}\eta +\eta \wedge \eta),\nonumber \\ \mathcal{F} & \longrightarrow \mathcal{F}, \nonumber \\
            \mathcal{G} & \longrightarrow \mathcal{G}+\mathcal{F}\wedge ^{\triangleright ^{\prime}} \eta.
        \end{align}
        When $H$ is abelian the term $\eta \wedge \eta$ is zero.
        The gauge transformation group is given by all pairs $(g,\eta)$, and the group product is given by the semi-direct product
        \begin{equation*}
            (g,\eta)(g^{\prime},\eta ^{\prime})=(gg^{\prime},(g\triangleright ^{\prime}\eta ^{\prime})\eta).
        \end{equation*}
\end{itemize}

\section{Physics from 2-connections}\label{sec:phys}
In this section we use categorical generalization to extend topological field theory of $BF$ fields giving rise to the $BFCG$ theory and use it to generalize the $BF$ coupling to gravity.
\subsection{BFCG}
The categorical generalization of the $BF$ theory is the topological theory denoted as $BFCG$ (see~\cite{girelli2008topological}). The theory is given by the action
\begin{equation}
        \label{bfcg}
        S_{BFCG}=\int_M \bigl \langle B\wedge \mathcal{F}\bigr \rangle_{\mathfrak{g}} + \bigl \langle C\wedge \mathcal{G}\bigr \rangle_{\mathfrak{h}},
    \end{equation}
where $B$ is a $\mathfrak{g}$-valued 2-form, $C$ is a $\mathfrak{h}$-valued 1-form and $\mathcal{F}$, $\mathcal{G}$ are the fake curvature and the 3-curvature, respectively. Additionally, $\bigl \langle ,\bigr \rangle_{\mathfrak{g}}$ and $\bigl \langle ,\bigr \rangle_{\mathfrak{h}}$ are bilinear, symmetric, non-degenerate, $G$-invariant, and invariant with respect to the Lie algebra commutator in their respective Lie algebras. The action $S_{BFCG}$, given by eq.\eqref{bfcg}, remains invariant under thin gauge transformations, eq.\eqref{tdelgadas}, if
\begin{align}
        B&\longrightarrow g^{-1}Bg,\nonumber \\
        C&\longrightarrow g^{-1}\triangleright ^{\prime \prime}C,
    \end{align}
    while invariance under fat gauge transformations, eq.~\eqref{tgruesas}, requires
    \begin{align}
        B&\longrightarrow B+C\wedge ^{\mathcal{T}}\eta,\nonumber \\
        C&\longrightarrow C.
    \end{align}
    To obtain the classical field equations, we utilize the principle of least action, leading us to
    \begin{equation}
\label{variacionbfcg}
    \begin{split}
        0 &= \delta \int_M \bigl \langle B\wedge \mathcal{F}\bigr \rangle_{\mathfrak{g}} + \bigl \langle C\wedge \mathcal{G}\bigr \rangle_{\mathfrak{h}}\\
        & = \int_M \bigl \langle \delta B\wedge \mathcal{F} \bigr \rangle_{\mathfrak{g}} +\bigl \langle B\wedge \delta \mathcal{F}\bigr \rangle_{\mathfrak{g}}+ \bigl \langle \delta C\wedge \mathcal{G}\bigr \rangle_{\mathfrak{h}}+ \bigl \langle C\wedge \delta \mathcal{G}\bigr \rangle_{\mathfrak{h}}\\
        & = \int_M \left(\bigl \langle \delta B\wedge \mathcal{F} \bigr \rangle_{\mathfrak{g}} +\bigl \langle \delta A\wedge d_{A}B \bigr \rangle_{\mathfrak{g}}+d\bigl \langle B\wedge \delta A\bigr \rangle_{\mathfrak{g}}-\bigl \langle \delta \beta \wedge \partial ^{\;\prime}(B) \bigr \rangle_{\mathfrak{h}}\right)\\[0.2cm]
         & \qquad {}  +[\bigl \langle \delta C\wedge \mathcal{G}\bigr \rangle_{\mathfrak{h}}+ \bigl \langle \delta \beta \wedge (dC+A\wedge ^{\triangleright ^{\prime}}C)\bigr \rangle_{\mathfrak{h}}-d\bigl \langle C\wedge \delta \beta\bigr \rangle_{\mathfrak{h}}\\[0.2cm]
         & \qquad {} +\bigl \langle \delta A\wedge (C\wedge ^{\mathcal{T}}\beta)\bigr \rangle_{\mathfrak{g}} ],
    \end{split}
\end{equation}
where we have utilized the formulas provided in Appendix \ref{app:a}. Thus, the action variation is zero for any variation $\delta A$, $\;\delta \beta$, $\;\delta B$, and $\;\delta C$ if and only if the following field equations are satisfied
\begin{align}
    d_{A}B+C\wedge ^{\mathcal{T}}\beta & =0,\nonumber \\ 
    dC+A\wedge ^{\triangleright ^{\prime}}C-\partial ^{\;\prime}(B) & =0, \nonumber \\ 
    \mathcal{F} & =0, \nonumber \\
    \mathcal{G} & =0,
\end{align}
where $\partial ^{\;\prime}:\mathfrak{g}\rightarrow \mathfrak{h}$ is the linear transformation defined by the rule\begin{equation}
       \bigl \langle \partial ^{\;\prime}(X), u\bigr \rangle _{\mathfrak{h}}= \bigl \langle X, \partial _{\ast}(u)\bigr \rangle _{\mathfrak{g}}\;\;\;\;\forall \;X\in \mathfrak{g}, \;\forall \;u\in \mathfrak{h}.
    \end{equation}
    \noindent Now, although the action~\eqref{bfcg} is invariant under thin gauge transformations~\eqref{tdelgadas} and fat gauge transformations \eqref{tgruesas}, its individual terms are not. To incorporate terms with powers of $B$ and $C$ invariant under both thin and fat gauge transformations into the $BFCG$ action is necessary to modify the action \eqref{bfcg} by introducing an auxiliary field $\alpha \in \Omega ^{1}(M,\mathfrak{h})$ ensuring invariance for each term in the action. This modification is known as the \emph{extended $BFCG$ action} (see~\cite{martins2011lie}) described by the action
    \begin{equation}
    \label{BFCG2}
        S_{BFCG2}=\int_M \bigl \langle B ^{\prime}\wedge \mathcal{F}\bigr \rangle_{\mathfrak{g}} + \bigl \langle C\wedge \mathcal{G}^{\prime}\bigr \rangle_{\mathfrak{h}},
    \end{equation}
    where $B^{\prime}$ has been written instead of $B$, as $B^{\prime}$ is chosen to be invariant under fat gauge transformations and 
    \begin{equation}
\mathcal{G}^{\prime}=\mathcal{G}+\mathcal{F}\wedge ^{\triangleright ^{\prime}}\alpha.
    \end{equation}
   \noindent It is worth mentioning that the prime in $B^{\prime}$ and $\mathcal{G}^{\prime}$ does not refer to a gauge transformation, but rather serves to distinguish them from $B$ and $\mathcal{G}$, respectively. The action~\eqref{BFCG2} is invariant under thin transformations if
    \begin{align}
    \label{tdelgadasb}
        \alpha & \longrightarrow g^{-1}\triangleright ^{\prime \prime}\alpha,\nonumber \\
        B^{\prime} & \longrightarrow g^{-1}B^{\prime}g,\nonumber \\
        C & \longrightarrow g^{-1}\triangleright ^{\prime \prime} C,
    \end{align} 
while invariance under fat gauge transformations \eqref{tgruesas}, requires
    \begin{align}
    \label{tgruesasb}
        \alpha & \longrightarrow \alpha -\eta,\nonumber \\
        B^{\prime} & \longrightarrow B^{\prime},\nonumber \\
        C & \longrightarrow C.
    \end{align}
Note that under fat gauge transformations, $\mathcal{F}$ and $\mathcal{G}^{\prime}$ are invariant, 
    \begin{equation}
        \mathcal{F} \longrightarrow \mathcal{F},\;\;\;\;\mathcal{G}^{\prime} \longrightarrow \mathcal{G}^{\prime}.
    \end{equation}
    Moreover, under thin gauge transformations, we have
        \begin{equation}
        \mathcal{F}\longrightarrow g^{-1}\mathcal{F}g,\;\;\;\; \mathcal{G}^{\prime}\longrightarrow g^{-1}\triangleright ^{\prime \prime}\mathcal{G}^{\prime}.
    \end{equation}
    The classical field equations of the extended $BFCG$ theory are derived by varying the action~\eqref{BFCG2}. To achieve this, we write
    \begin{equation*}
        \begin{split}
            S_{BFCG2}&=\int_M \bigl \langle B ^{\prime}\wedge \mathcal{F}\bigr \rangle_{\mathfrak{g}} + \bigl \langle C\wedge \mathcal{G}\bigr \rangle_{\mathfrak{h}} + \bigl \langle C\wedge (\mathcal{F}\wedge ^{\triangleright ^{\prime}}\alpha)\bigr \rangle_{\mathfrak{h}}\\[0.2cm]
            & =S_{B^{\prime}\mathcal{F}}+S_{CG}+\int_M \bigl \langle C\wedge (\mathcal{F}\wedge ^{\triangleright ^{\prime}}\alpha)\bigr \rangle_{\mathfrak{h}}.\\[0.2cm]
        \end{split}
    \end{equation*}
    Using eqs.~\eqref{variacionbfcg},~\eqref{identidadBdeltaF} and the relation
    $$
        \bigl \langle C\wedge (\mathcal{F}\wedge ^{\triangleright ^{\prime}}\alpha)\bigr \rangle_{\mathfrak{h}}=-\bigl \langle (C\wedge ^{\mathcal{T}} \alpha)\wedge \mathcal{F}\bigr \rangle_{\mathfrak{g}}=-\bigl \langle (\alpha \wedge ^{\mathcal{T}} C)\wedge \mathcal{F}\bigr \rangle_{\mathfrak{g}},
 $$ 
 requesting the variation of the action to be equal to zero translates to
\begin{align}
            0 &= \delta \int_M \bigl \langle B^{\prime}\wedge \mathcal{F}\bigr \rangle_{\mathfrak{g}} + \bigl \langle C\wedge \mathcal{G}^{\prime}\bigr \rangle_{\mathfrak{h}} \nonumber \\[0.2cm]
            & = \int_M \bigl \langle \delta B^{\prime}\wedge \mathcal{F} \bigr \rangle_{\mathfrak{g}} +\bigl \langle \delta A\wedge d_{A}B^{\prime} \bigr \rangle_{\mathfrak{g}}+d\bigl \langle B^{\prime}\wedge \delta A\bigr \rangle_{\mathfrak{g}}-\bigl \langle \delta \beta \wedge \partial ^{\;\prime}(B^{\prime}) \bigr \rangle_{\mathfrak{h}}\nonumber\\[0.2cm]
            & \qquad {}\qquad {} +\bigl \langle \delta C\wedge \mathcal{G}\bigr \rangle_{\mathfrak{h}}+ \bigl \langle \delta \beta \wedge (dC+A\wedge ^{\triangleright ^{\prime}}C)\bigr \rangle_{\mathfrak{h}}-d\bigl \langle C\wedge \delta \beta\bigr \rangle_{\mathfrak{h}}\nonumber\\[0.2cm]
            & \qquad {}\qquad {}+\bigl \langle \delta A\wedge (C\wedge ^{\mathcal{T}}\beta)\bigr \rangle_{\mathfrak{g}} \nonumber\\[0.2cm]
            & \qquad {}\qquad {} + \bigl \langle \delta C\wedge (\mathcal{F}\wedge ^{\triangleright ^{\prime}}\alpha)\bigr \rangle_{\mathfrak{h}}+ \bigl \langle \delta \alpha \wedge (\mathcal{F}\wedge ^{\triangleright ^{\prime}}C) \bigr \rangle_{\mathfrak{h}}-d\bigl \langle (C\wedge ^{\mathcal{T}} \alpha)\wedge \delta A\bigr \rangle_{\mathfrak{g}}\nonumber\\[0.2cm]
            & \qquad {}\qquad {} -\bigl \langle \delta A\wedge d_{A}(C\wedge ^{\mathcal{T}} \alpha) \bigr \rangle_{\mathfrak{g}}+ \bigl \langle \delta \beta \wedge \partial ^{\;\prime}(C\wedge ^{\mathcal{T}} \alpha) \bigr \rangle_{\mathfrak{h}},
    \end{align}
    from which we obtain the classical field equations
    \begin{align}
        d_{A}(B^{\prime}-C\wedge ^{\mathcal{T}} \alpha)+C \wedge ^{\mathcal{T}}\beta & =0,
    \nonumber \\
        dC+A\wedge ^{\triangleright ^{\prime}}C-\partial ^{\prime}(B^{\prime}-C\wedge ^{\mathcal{T}} \alpha) & =0,
    \nonumber \\
        \mathcal{F} & =0,
    \nonumber \\
\mathcal{G}^{\prime} =\mathcal{G}+\mathcal{F}\wedge ^{\triangleright ^{\prime}}\alpha & =0,
    \nonumber \\
        \mathcal{F}\wedge ^{\triangleright ^{\prime}}C & =0.
    \end{align}
    which are obtained from the variation of $A$, $\beta$, $B$, $C$, and $\alpha$, respectively.
    
\subsection{BFCG sequestered gravity}
In this section, couplings with and without kinetic terms are introduced, along with two volume elements. One of these volume elements is exclusive to $BFCG$. The $BFCG$ action coupled with gravity that we will consider is a generalization of the coupling between gravity and $BF$ theory (see~\cite{alexander2020topological}) in which an \emph{extended $BFCG$ action} is employed,
\begin{align}
\label{accionbfcggravedad1}
	S & =\int_M \bigl \langle B^{\prime}\wedge \mathcal{F} \bigr \rangle_{\mathfrak{g}}+\bigl \langle C\wedge \mathcal{G}^{\prime}\bigr \rangle_{\mathfrak{h}}+ \left[\frac{1}{2\kappa}R(\hat{g})-\frac{\bar{\Lambda}}{\kappa}+\mathscr{L}_{M}\right] \;\bigl \langle B^{\prime}\wedge B^{\prime}\bigr \rangle_{\mathfrak{g}} \nonumber \\
& =S_{B^{\prime}F}+S_{CG^{\prime}}+S_{RG},
\end{align}
where $\kappa=8\pi G c^{-4}$ and
\begin{equation}
    \;\bigl \langle B^{\prime}\wedge B^{\prime}\bigr \rangle_{\mathfrak{g}} \equiv \frac{1}{4}d^{4}x\epsilon^{\mu\nu\rho\sigma}(B^{\prime})_{\mu \nu}^{a}(B^{\prime})_{\rho \sigma}^{b}Q_{ab}=d^{4}x\;\sqrt{\hat{g}}\neq 0,
\end{equation} 
is an invariant volume form under both thin and fat gauge transformations \eqref{tdelgadasb}, \eqref{tgruesasb}.\\
The composite metric $\hat{g}_{\mu \nu}$ is defined as
\begin{equation}
    \hat{g}_{\mu \nu}=(\frac{\hat{g}}{g})^{\frac{1}{4}}\;g_{\mu \nu},
\end{equation}
where $g_{\mu \nu}$ is an arbitrary metric, considered as one of the fundamental fields in the action and $g$ its determinant. The energy-momentum tensor is
\begin{equation}
    T_{\mu \nu}(\hat{g})=\frac{-2}{\sqrt{\hat{g}}}\frac{\delta(\sqrt{\hat{g}}\mathscr{L}_{M})}{\delta \hat{g}^{\mu \nu}}=-2\frac{\delta\mathscr{L}_{M}}{\delta \hat{g}^{\mu \nu}}+\hat{g}_{\mu \nu}\mathscr{L}_{M},
\end{equation}
and its trace is given by
\begin{equation}
    T=T_{\mu \nu}\hat{g}^{\mu \nu}.
\end{equation}
The fields on which the action~\eqref{accionbfcggravedad1} depends are: $A, \beta, B^{\prime}, C$, $\alpha$, and $g^{\mu \nu}$.
Since neither $S_{B^{\prime}F}$ nor $S_{CG^{\prime}}$ depends on $g^{\mu \nu}$, the variation of \eqref{accionbfcggravedad1} with respect to $g^{\mu \nu}$ leads to the trace-free Einstein equation,
\begin{equation}
\label{primeraecuacionbfcgsecuestered}
    R_{\mu \nu }(\hat{g})-\frac{1}{4}R(\hat{g})\hat{g}_{\mu \nu }=\kappa \left(T_{\mu \nu}-\frac{1}{4}T\hat{g}_{\mu \nu}\right),
\end{equation}
as shown in eq.~\eqref{desarrollovariaciónrg}. Also, since $S_{CG^{\prime}}$ does not depend on $B^{\prime}$, the variation of the action with respect to $B^{\prime}$ leads to
\begin{equation}
\label{segundaecuacionbfcgsecuestered}
    \mathcal{F}+\frac{1}{2\kappa}\left[R(\hat{g})+\kappa T -4\bar{\Lambda}\right]B^{\prime}=0,
\end{equation}

\noindent as shown in eq.~\eqref{primeraecuacionbfsecuestered}. The variations with respect to $A, \beta, C$, and $\alpha$ are straightforward
\begin{align}
\label{primis}
     d_{A}(B^{\prime}-C\wedge ^{\mathcal{T}} \alpha)+C \wedge ^{\mathcal{T}}\beta &=0,
\\
\label{seguns}
    dC+A\wedge ^{\triangleright ^{\prime}}C-\partial ^{\prime}(B^{\prime}-C\wedge ^{\mathcal{T}} \alpha) &=0,
\\
\label{tercis}
    \mathcal{G}^{\prime}=\mathcal{G}+\mathcal{F}\wedge ^{\triangleright ^{\prime}}\alpha &=0,
\\
\label{cuartas}
    \mathcal{F}\wedge ^{\triangleright ^{\prime}}C &=0.
\end{align}
\begin{Example}
Let $\mathfrak{X}$ be the crossed module defined as 
in Example \ref{modulocruzadoabeliano}, we find $\partial _{\ast}$ is a isomorphism, $\partial ^{\prime}=\partial_{\ast}^{-1}$, $\triangleright ^{\prime}=0$ and $\mathcal{T}=0$. Thus, eqs.~\eqref{segundaecuacionbfcgsecuestered}, \eqref{primis}, \eqref{seguns}, \eqref{tercis} and \eqref{cuartas} simplify to
    \begin{align*}
        F-\partial_{\ast}( \beta) +\frac{1}{2\kappa}\left[R(\hat{g})+\kappa T -4\bar{\Lambda}\right]B^{\prime}= & 0,
    \\
        dB^{\prime}= &0,
    \\
        dC= &\partial_{\ast}^{-1}B^{\prime},
    \\
        d\beta = &0,
    \end{align*}
    since eq.~\eqref{cuartas} with $\triangleright ^{\prime}=0$ is trivial.
    Applying the exterior derivative to the first of them and using the other relations along with the Bianchi identity, $dF=ddA=0$, yields
    \begin{equation}
        d(R+\kappa T)\wedge B^{\prime}=0,
    \end{equation}
   since $\sqrt{\hat{g}}\neq 0$, we have (see~\cite{alexander2020topological})
    \begin{equation}
        d(R+\kappa T)=0,
    \end{equation}
    from where the trace-free Einstein equation can be rewritten as
 \begin{equation}
        R_{\mu \nu}(\hat{g})-\frac{1}{2}\hat{g}_{\mu \nu}R(\hat{g})+\Lambda \hat{g}_{\mu \nu}=\kappa T_{\mu \nu},
    \end{equation}
    where $\Lambda$ is a constant of integration. This shows that for a particular case of a strict $2$-group of abelian type, the generalized $BF$ theory coupled with gravity reduces to the $BF$ theory coupled with gravity reported in the literature and developed following the action in eq.~\eqref{eq:seq}.
\end{Example}
\subsection{Coupling with the form C}
    When working with the $BFCG$ theory, a possibility for coupling gravity arises that was not present in the $BF$ theory, which involves constructing the volume form from $C$, leading to the formulation of the action
    \begin{equation}
\label{accionbfcggravedad2}
	S=\int_M \bigl \langle B^{\prime}\wedge \mathcal{F} \bigr \rangle_{\mathfrak{g}}+\bigl \langle C\wedge \mathcal{G}^{\prime}\bigr \rangle_{\mathfrak{h}}+ \left[\frac{1}{2\kappa}R(\hat{g})-\frac{\bar{\Lambda}}{\kappa}+\mathscr{L}_{M}\right] \;d^{4}x\;\sqrt{\hat{g}}=S_{B^{\prime}F}+S_{CG^{\prime}}+S_{RG},
\end{equation}
where
\begin{equation}
     d^{4}x\;\sqrt{\hat{g}}= \;\bigl \langle (C\wedge ^{\mathcal{T}} C)\wedge (C\wedge ^{\mathcal{T}} C)\bigr \rangle_{\mathfrak{h}}\neq 0, 
\end{equation} 
is a volume form. In other words, instead of the term $\langle B^{\prime} \wedge B ^{\prime}\rangle_{\mathfrak{g}}$, the term $ \langle (C\wedge ^{\mathcal{T}} C)\wedge (C\wedge ^{\mathcal{T}} C) \rangle_{\mathfrak{h}}$ is used. The composite metric, $\hat{g}_{\mu \nu}$, is defined in the same way as
\begin{equation*}
    \hat{g}_{\mu \nu}=(\frac{\hat{g}}{g})^{\frac{1}{4}}\;g_{\mu \nu},
\end{equation*}
where $g_{\mu \nu}$ is an arbitrary metric. The fields on which the action~\eqref{accionbfcggravedad2} depends are: $A, \beta, B^{\prime}, C, \alpha$, and $g^{\mu \nu}$. Since neither $S_{B^{\prime}F}$ nor $S_{CG^{\prime}}$ depends on $g^{\mu \nu}$, the variation with respect to $g^{\mu \nu}$ of the action~\eqref{accionbfcggravedad2} leads to the trace-free Einstein equation,
\begin{equation}
\label{primeraecuacionbfcgsecuestered2}
    R_{\mu \nu }(\hat{g})-\frac{1}{4}R(\hat{g})\hat{g}_{\mu \nu }=\kappa \left(T_{\mu \nu}-\frac{1}{4}T\hat{g}_{\mu \nu}\right),
\end{equation}
as shown in eq.~\eqref{desarrollovariaciónrg}. The variations of the action with respect to $C, A, \beta, B^{\prime}$, and $\alpha$ lead, respectively, to
\begin{align}
\label{segundaecuacionbfcgsecuestered2}
    \mathcal{G}^{\prime}+\frac{1}{\kappa}\left[R(\hat{g})+\kappa T -4\bar{\Lambda}\right]\left[(C\wedge ^{\mathcal{T}}C)\wedge ^{\triangleright ^{\prime}}C\right] &=0.
\\
    d_{A}(B^{\prime}-C\wedge ^{\mathcal{T}} \alpha)+C \wedge ^{\mathcal{T}}\beta &=0,
\\
     dC+A\wedge ^{\triangleright ^{\prime}}C-\partial ^{\prime}(B^{\prime}-C\wedge ^{\mathcal{T}} \alpha) & =0,
\\
    \mathcal{F} &=0,
\\
        \mathcal{F}\wedge ^{\triangleright ^{\prime}}C &=0.
\end{align}
These equations detach the cosmological constant from the curvature $\mathcal{F}$ and instead establish a connection with the $3$-curvature $\mathcal{G^{\prime}}$, characteristic of higher norm theories. This opens possibilities for developing a fully generalized theory that integrates both the $B^{\prime}$ and $C$ fields as well as for researching such generalized curvatures and their implications within this theory. For instance, the new volume forms available in higher gauge theories may have an application in the study of bigravity, massive gravity and multimetric gravity theories~\cite{Baldacchino:2016jsz, deRham:2014zqa, Hassan:2011zd}. This possibility will be explored elsewhere.
\subsection{Couplings with kinetic term}
In this section we consider a minimal extension of the $BFCG$ theory coupled to gravity, which includes a kinetic term (a modified version of the Higher Yang-Mills term from~\cite{baez2002higher}, which remains invariant under gauge transformations) and maintains invariance under thin gauge transformations~\eqref{tdelgadas} and fat gauge transformations~\eqref{tgruesas} of the gauge fields $A$ and $\beta$. The action for this extension is given by
\begin{equation}
\label{eq:2ym}
    S=S_{0}+\frac{1}{4e^{2}}\int_M \bigl \langle \mathcal{F}\wedge \hat{\star} \mathcal{F}\bigr \rangle _{\mathfrak{g}} +\bigl \langle \mathcal{G}^{\prime}\wedge \hat{\star }\mathcal{G}^{\prime}\bigr \rangle _{\mathfrak{h}}=S_{0}+S_{2YM},
\end{equation}
where $e$ is the gauge coupling, $S_{0}$ is any of the actions~\eqref{accionbfcggravedad1} (with the volume form given through $B$) or~\eqref{accionbfcggravedad2} (with the volume form given through $C$), $\hat{\star}$ is the Hodge star operator associated with the composite metric $\hat{g}_{\mu \nu}$ and
\begin{equation}     \mathcal{G}^{\prime}=\mathcal{G}+\mathcal{F}\wedge ^{\triangleright ^{\prime}}\alpha.
\end{equation}
\noindent The variation of $S_{2YM}$ while keeping the fields $B^{\prime}$ and $g_{\mu \nu}$ (and therefore $\hat{g}_{\mu \nu}$) fixed is
\begin{equation}
    \begin{split}
        \delta ^{\prime}S_{2YM} & =\frac{1}{2e^{2}}\int_M \bigl \langle \delta \mathcal{F}\wedge \hat{\star} \mathcal{F}\bigr \rangle _{\mathfrak{g}} -\bigl \langle \hat{\star }\mathcal{G}^{\prime}\wedge  \delta \mathcal{G}\bigr \rangle _{\mathfrak{h}}-\bigl \langle \hat{\star }\mathcal{G}^{\prime}\wedge  \delta (\mathcal{F}\wedge ^{\triangleright ^{\prime}}\alpha)\bigr \rangle _{\mathfrak{h}}\\[0.2cm]
        & =\frac{1}{2e^{2}}\int_M d\bigl \langle \hat{\star} \mathcal{F}\wedge \delta A\bigr \rangle_{\mathfrak{g}} +\bigl \langle \delta A\wedge d_{A}(\hat{\star} \mathcal{F})\bigr \rangle_{\mathfrak{g}}-\bigl \langle \delta \beta \wedge \partial ^{\;\prime}(\hat{\star} \mathcal{F}) \bigr \rangle_{\mathfrak{h}}+\\[0.2cm]
        & \qquad {} \qquad {}\; +d\bigl \langle \hat{\star }\mathcal{G}^{\prime}\wedge \delta \beta \bigr \rangle_{\mathfrak{h}}-\bigl \langle \delta \beta \wedge (d\hat{\star }\mathcal{G}^{\prime}+A\wedge ^{\triangleright ^{\prime}}(\hat{\star }\mathcal{G}^{\prime})) \bigr \rangle_{\mathfrak{h}}\\[0.2cm]
        & \qquad {} \qquad {}\;-\bigl \langle \delta A\wedge (\hat{\star }\mathcal{G}^{\prime}\wedge ^{\mathcal{T}}\beta) \bigr \rangle_{\mathfrak{g}}+d\bigl \langle (\alpha \wedge ^{\mathcal{T}}\hat{\star} \mathcal{G}^{\prime})\wedge \delta A\bigr \rangle_{\mathfrak{g}}\\[0.2cm]
        &\qquad {} \qquad {}\; +\bigl \langle \delta A\wedge d_{A}(\alpha \wedge ^{\mathcal{T}}\hat{\star} \mathcal{G}^{\prime})\bigr \rangle_{\mathfrak{g}}-\bigl \langle \delta \beta \wedge \partial ^{\;\prime}(\alpha \wedge ^{\mathcal{T}}\hat{\star} \mathcal{G}^{\prime}) \bigr \rangle_{\mathfrak{h}}\\[0.2cm]
        &\qquad {} \qquad {}\;+\bigl \langle \delta \alpha \wedge (\mathcal{F}\wedge ^{\triangleright ^{\prime}}\hat{\star} \mathcal{G}^{\prime}) \bigr \rangle_{\mathfrak{h}}.
    \end{split}
\end{equation}
Where we have used the equations from Appendix \ref{app:a} and the equation
\begin{equation*}
\delta \bigl \langle \mathcal{R}\wedge \hat{\star} \mathcal{R}\bigr \rangle=2\bigl \langle \delta \mathcal{R}\wedge \hat{\star} \mathcal{R}\bigr \rangle,
\end{equation*}
which holds for variations where 
$B^{\prime}$, $g_{\mu \nu}$ and therefore $\hat{g^{\mu \nu}}$ remain fixed. On the other hand, using local expressions in components for $\mathcal{F}$, $\hat{\star}\mathcal{F}$, $\mathcal{G}^{\prime}$ and $\hat{\star}\mathcal{G}^{\prime}$, as well as properties of  Levi-Civita symbols and generalized Kronecker delta, the variation of $S_{2YM}$ with respect to $\hat{g^{\mu \nu}}$ (therefore, indirectly, with respect to $B^{\prime}$ and $g^{\mu \nu}$) is
\begin{equation}
        \frac{\delta S_{2YM}}{\delta \hat{g}^{\mu \nu }}\delta \hat{g}^{\mu \nu }
         =\int_M d^{4}x\frac{\sqrt{\hat{g}}}{2\kappa}\;(- \kappa )\left[T_{\mu  \nu}^{(\mathcal{F})}+T_{\mu \nu}^{(\mathcal{G}^{\prime})}\right]\delta \hat{g}^{\mu \nu},
\end{equation}
\noindent where the energy-momentum tensors associated with the gauge fields $\mathcal{F}$, $\mathcal{G}^{\prime}$ are
\begin{align}
    T_{\mu \nu}^{(\mathcal{F})} & =-\frac{Q_{ab}}{2 e^{2}}\left[\mathcal{F}_{\alpha \mu}^{a}\;\mathcal{F}_{\beta \nu}^{b}\;\hat{g}^{\alpha \beta}-\frac{1}{4}\mathcal{F}_{\alpha \beta}^{a}\;\mathcal{F}_{\gamma \delta}^{b}\;\hat{g}^{\alpha \gamma}\;\hat{g}^{\beta \delta}\;\hat{g}_{\mu \nu}\right],
\\
    T_{\mu \nu}^{(\mathcal{G}^{\prime})} & =-\frac{q_{\lambda \xi}}{4 e^{2}}\left[(\mathcal{G}^{\prime})_{\mu \alpha \beta}^{\lambda}(\mathcal{G}^{\prime})_{\nu \gamma \delta}^{\xi}\;\hat{g}^{\alpha \gamma}\;\hat{g}^{\beta \delta}-\frac{1}{6}(\mathcal{G}^{\prime})_{\mu ^{\prime}\alpha \beta}^{\lambda}(\mathcal{G}^{\prime})_{\nu ^{\prime}\gamma \delta}^{\xi}\;\hat{g}^{\mu ^{\prime}\nu ^{\prime}}\;\hat{g}^{\alpha \gamma}\;\hat{g}^{\beta \delta}\hat{g}_{\mu \nu}\right].
\end{align}
\noindent It is easy to see that
\begin{equation*}
    T^{(\mathcal{F})}=T_{\mu \nu}^{(\mathcal{F})}\hat{g}^{\mu \nu}=0,
\end{equation*}
and 
\begin{equation}
    T_{\mu \nu}=T_{\mu \nu}^{(M)}+T_{\mu \nu}^{(\mathcal{F})}+T_{\mu \nu}^{(\mathcal{G}^{\prime})}= \frac{-2}{\sqrt{\hat{g}}}\frac{\delta(\sqrt{\hat{g}}\mathscr{L}_{M\mathcal{F}\mathcal{G}^{\prime}})}{\delta \hat{g}^{\mu \nu}}\equiv -2\frac{\delta\mathscr{L}_{M\mathcal{F}\mathcal{G}^{\prime}}}{\delta \hat{g}^{\mu \nu}}+\hat{g}_{\mu \nu}\mathscr{L}_{M\mathcal{F}\mathcal{G}^{\prime}},
\end{equation}
where $\mathscr{L}_{M\mathcal{F}\mathcal{G}^{\prime}}$ is the Lagrangian density of the matter field and the gauge fields provided by the kinetic terms. Finally, the field equations are
\begin{itemize}
    \item For the action~\eqref{accionbfcggravedad1} (volume form given through $B^{\prime}$), the variations with respect to $g^{\mu \nu}, B^{\prime}, C, A, \beta$ and $\alpha$ are, respectively,
    \begin{align}
        R_{\mu \nu }(\hat{g})-\frac{1}{4}R(\hat{g})\hat{g}_{\mu \nu } - \kappa \left(T_{\mu \nu}-\frac{1}{4}T\hat{g}_{\mu \nu}\right) & = 0,
    \nonumber \\
        \mathcal{F}+\frac{1}{2\kappa}\left[R(\hat{g})+\kappa T -4\bar{\Lambda}\right]B^{\prime} & =0,
     \nonumber\\
         \mathcal{G}+\mathcal{F}\wedge ^{\triangleright ^{\prime}}\alpha & =0,
    \nonumber\\
        d_{A}(B^{\prime}-C\wedge ^{\mathcal{T}} \alpha)+C\wedge ^{\mathcal{T}}\beta+\frac{1}{2e^{2}}d_{A}(\hat{\star }\mathcal{F}) & =0,
    \nonumber\\
        dC+A\wedge ^{\triangleright ^{\prime}}C-\partial ^{\prime}(B^{\prime}-C\wedge ^{\mathcal{T}} \alpha)-\frac{1}{2e^{2}}\partial ^{\prime}(\hat{\star}\mathcal{F}) & =0,
     \nonumber\\
        \mathcal{F}\wedge ^{\triangleright ^{\prime}}C& =0.  \label{potencialdadoporb}
    \end{align}
    \item For the action~\eqref{accionbfcggravedad2} (volume form given through $C$), the variations with respect to $g^{\mu \nu}, C, B^{\prime}, A, \beta$ and $\alpha$ are, respectively,
    \begin{align}
        R_{\mu \nu }(\hat{g})-\frac{1}{4}R(\hat{g})\hat{g}_{\mu \nu } - \kappa \left(T_{\mu \nu}-\frac{1}{4}T\hat{g}_{\mu \nu}\right) & = 0,
    \nonumber \\
        \mathcal{G}^{\prime}+\frac{1}{\kappa}\left[R(\hat{g})+\kappa T -4\bar{\Lambda}\right]\left[(C\wedge ^{\mathcal{T}}C)\wedge ^{\triangleright ^{\prime}}C\right] &=0,
    \nonumber \\
        \mathcal{F} & =0,
    \nonumber \\
        d_{A}(B^{\prime}-C\wedge ^{\mathcal{T}} \alpha)+C\wedge ^{\mathcal{T}}\beta &\nonumber \\ 
        -\frac{1}{2e^{2}}\left[\hat{\star}\mathcal{G}^{\prime}\wedge ^{\mathcal{T}}\beta-d_{A}(\alpha \wedge ^{\mathcal{T}}\hat{\star}\mathcal{G}^{\prime})\right] & =0,
    \nonumber \\
        dC+A\wedge ^{\triangleright ^{\prime}}C-\partial ^{\prime}(B^{\prime}-C\wedge ^{\mathcal{T}} \alpha)& \nonumber \\
        -\frac{1}{2e^{2}}\left[d(\hat{\star}\mathcal{G}^{\prime})+A\wedge ^{\triangleright ^{\prime}}(\hat{\star}\mathcal{G}^{\prime})+\partial ^{\prime}(\alpha \wedge ^{\mathcal{T}}\hat{\star}\mathcal{G}^{\prime})\right] & =0,
    \nonumber \\
        \mathcal{F}\wedge ^{\triangleright ^{\prime}}C+\frac{1}{2e^{2}}(\mathcal{F}\wedge ^{\triangleright ^{\prime}}\hat{\star} \mathcal{G}^{\prime}) & =0.    \label{potencialdadoporc}
    \end{align}
\end{itemize}
\begin{Example}
Let $\mathfrak{X}$ be the crossed module defined as 
in Example \ref{modulocruzadoabeliano}, then $\partial _{\ast}$ is an isomorphism, $\partial ^{\prime}=\partial_{\ast}^{-1}$, $\triangleright ^{\prime}=0$ and $\mathcal{T}=0$. Thus, the field eqs.~\eqref{potencialdadoporb} are reduced to    
\begin{equation*}
            F-\partial_{\ast}(\beta)+\frac{1}{2\kappa}\left[R(\hat{g})+\kappa T -4\bar{\Lambda}\right]B^{\prime}=0,
        \end{equation*}
        \begin{equation*}
            \mathcal{G}^{\prime}=d\beta =0,
        \end{equation*}
        \begin{equation*}
            dB^{\prime}=-\frac{1}{2e^{2}}d(\hat{\star}(F-\partial _{\ast}(\beta))),
        \end{equation*}
        \begin{equation*}
            dC-\partial_{\ast}^{-1}(B^{\prime})=\frac{1}{2e^{2}}\partial_{\ast}^{-1}(\hat{\star}(F-\partial_{\ast}(\beta))).
        \end{equation*}
        From the previous equations, it follows that
         \begin{equation}
        \label{eq:bfsecuesteredu1}
            \frac{d\Lambda \wedge \mathcal{F}}{(\Lambda-\bar{\Lambda})^{2}}=-\frac{1}{2}\left(\frac{2}{e^{2}\kappa}d\hat{\star}\mathcal{F}\right)
        \end{equation}
        where
         \begin{equation*}
            \Lambda =\frac{1}{4}(R(\hat{g})+\kappa T),\;\;\; \mathcal{F}=F-\beta.
        \end{equation*}
         Equation~\eqref{eq:bfsecuesteredu1}, along with
         \begin{equation}
            T_{\mu \nu}^{(\mathcal{F})}=\frac{1}{e^{2}}\left[\mathcal{F}_{\alpha \mu}\;\mathcal{F}_{\beta \nu}\;\hat{g}^{\alpha \beta}-\frac{1}{4}\mathcal{F}_{\alpha \beta}\;\mathcal{F}_{\gamma \delta}\;\hat{g}^{\alpha \gamma}\;\hat{g}^{\beta \delta}\;\hat{g}_{\mu \nu}\right],
        \end{equation}
        are equivalent, up to reparametrization of $e$, to the abelian equations presented in~\cite{alexander2022black} where it is shown that this model leads to astrophysical solutions that, in some limits, resemble Reissner-Nordstr\"om black holes (solutions to the Einstein-Maxwell theory) but introduce the notion of a fundamental unit charge. This further demonstrates that, for a specific case of a strict 2-group abelian type, the generalized $BF$ theory coupled with gravity including a kinetic term, reduces to the $BF$ theory coupled with gravity including a kinetic term, reported in the literature. However, in the non-abelian case, it is speculated that additional information arises due to the nontrivial structure of the 2-connection, leading to a modifications to the non-abelian case of BF coupled to gravity. It would be interesting to study such corrections in the context of cosmological solutions, and even more, to take advantage of the richer gauge symmetry structure of higher gauge theory in order to combine gauge fields that are known to be relevant, on the one hand, for astrophysical solutions, and on the other hand, for cosmology.
\end{Example}

\section{Conclusions
}\label{sec:con}

In the first part of this work, we provided a basic introduction to the theory of $G$-principal bundles and its generalization to higher gauge theory.
After reviewing these tools, we moved on to work out a physical application original to this paper: the coupling of categorically generalized $BF$ theory with gravity, including a generalized kinetic term. We derived the equations of motion for an arbitrary crossed module. 
As a result of the generalized coupling, we extended the field content of unimodular gravity while maintaining it as a background-independent theory. However, as expected, the new theory is not topological, independently of whether it has an explicit kinetic term. The fields in the theory depend on the Lie groups used to define the crossed module. For example, for a $U(1),\; SO(2)$ (crossed module) gauge fields can be coupled to unimodular gravity, and by adding the kinetic term, this coupling is reflected in the metric field equations through the addition of the electromagnetic energy-momentum tensor. In addition, higher-order categories introduce new possibilities for the volume form and, therefore, for the action that couples BFCG and gravity. At the level of the field equations, this determines whether the gravitational part of the action acts as a source for $\mathcal{F}$ or for 
$\mathcal{G}$. This is expected to modify the space of solutions of the theory, as will be reported elsewhere.\\

In the specific case where the crossed module is constructed out of $U(1)$ and $SO(2)$ groups, we show that the standard $BF$ theory coupled with gravity is recovered. However, in the non-abelian case, it is speculated that additional information arises due to the nontrivial structure of the 2-connection, leading to theory that differs from those found in the existing literature. Finally, it is worth emphasizing that the resulting theory is not simply unimodular gravity minimally coupled to electromagnetism, since there are further equations arising from the BF part of the action and from the coupling between the 2-form  $B$ and the Einstein-Hilbert term.\\

Furthermore, there are new gauge transformations. In particular, in the context of $2$-connections, we identified thin and fat gauge transformations. The first are the usual gauge transformations plus a transformation for the new $2$-form connection that appears in the theory. The second are entirely new transformations that are highly useful due to their versatility. For example, in different contexts, they can be applied to develop a formalism in which the transformations of the fields are the same as those of the fields present in the $BFCG$ theory, demonstrating the practical utility of fat transformations. For instance, for the $2$-Poincare group the transformation properties of the one-form $C$ are the same as the transformation properties of the tetrad one-form $e$ under the local Lorentz and the diffeomorphism transformations \cite{mikovic2012poincare}. In future work, we plan to use the versatility of fat transformations
to explore cosmological solutions for the $BFCG$ theory coupled with gravity proposed here.

\appendix

\section{Remarks on differential crossed modules}
\label{app:modulocruzadodiferencial}
\noindent The differential crossed module $(\mathfrak{g},\mathfrak{h},\partial _{\ast}, \triangleright ^{\prime})$ associated with $(G,H,\partial ,\triangleright)$ satisfies~\cite{martins2011lie}:
\begin{enumerate}
\item For any $X \in \mathfrak{g}$, the function $\phi_{X}^{\prime}:\mathfrak{h}\rightarrow \mathfrak{h}$ defined as $\phi_{X}^{\prime}(\xi)=X\triangleright ^{\prime} \xi$ for any $\xi \in \mathfrak{h}$, is a derivation of $\mathfrak{h}$, which can be expressed as
\begin{equation*}
            X\triangleright ^{\prime} [\xi,\nu]=[X\triangleright ^{\prime}\xi,\nu]+[\xi ,X\triangleright ^{\prime} \nu],\;\;\;\; \forall \;X\in \mathfrak{g},\;\forall \;\xi, \nu \in \mathfrak{h}.
        \end{equation*}
\item The function $\triangleright ^{\prime}:\mathfrak{g}\rightarrow Der(\mathfrak{h})$ from the Lie algebra $\mathfrak{g}$ to the algebra of derivations of $\mathfrak{h}$ is a Lie algebra morphism, which can be written as
\begin{equation*}
            [X,Y]\triangleright^{\prime} \xi=X\triangleright ^{\prime} (Y\triangleright ^{\prime} \xi)-Y\triangleright ^{\prime}(X\triangleright ^{\prime} \xi),\;\;\;\; \forall \;X,Y\in \mathfrak{g},\;\forall \;\xi\in \mathfrak{h}.
        \end{equation*}
\item $\partial _{\ast}(X\triangleright ^{\prime} \xi)=[X,\partial _{\ast} (\xi)],\;\;\;\; \forall \;X\in \mathfrak{g},\;\forall \;\xi\in \mathfrak{h}.$
\item $\partial _{\ast}(\xi)\triangleright ^{\prime} \nu =[\xi,\nu]\;\;\;\;\forall \;\xi\in \mathfrak{h}.$
\end{enumerate}
\noindent An identity useful for proving the norm invariance of the extended $BFCG$ action is
\begin{equation}
        \label{invarianzabfcgextendida}
        g\triangleright ^{\prime \prime}(X\triangleright ^{\prime}\xi)=(gXg^{-1})\triangleright ^{\prime}(g\triangleright ^{\prime \prime}\xi),\;\;\;\; \forall g\in G,\;\xi \in \mathfrak{h},\;X\in \mathfrak{g},
    \end{equation}
where $\begin{tikzcd}
                \triangleright ^{\prime \prime}:G \arrow[r,"\triangleright"] & Aut(H) \arrow[r,"\ast"] & GL(\mathfrak{h})
            \end{tikzcd}$ is defined as $
        \triangleright ^{\prime \prime} (g)=(\phi _{g})_{\ast}, \forall g\in G$, with $\phi _{g}(h)=g\triangleright h, \forall h\in H$ and $(\phi _{g})_{\ast}$ denotes its push-forward. Additionally,  $g\triangleright ^{\prime \prime}\nu$ is defined as
     $ g\triangleright ^{\prime \prime}\nu=\triangleright ^{\prime \prime}(g)\nu$
    
\begin{Definition}
\label{formasbilinealssimetricasinvariantes}
   Let $\mathcal{X}=(G,H,\partial ,\triangleright)$ be a crossed module and let $\bar{\mathcal{X}}=(\mathfrak{g},\mathfrak{h},\partial _{\ast}, \triangleright ^{\prime})$ be the associated differential crossed module. A bilinear, symmetric, non-degenerate, $G$-invariant, and invariant form on $\bar{\mathcal{X}}$ is a pair of bilinear, symmetric, non-degenerate forms, $ \langle , \rangle_{\mathfrak{g}}$ on $\mathfrak{g}$ and $ \langle ,\rangle_{\mathfrak{h}}$ on $\mathfrak{h}$ such that:
\begin{enumerate}
\item $ \langle , \rangle_{\mathfrak{g}}$ is $G$-invariant:
$
            \bigl \langle gXg^{-1},gYg^{-1}\bigr \rangle_{\mathfrak{g}}=\bigl \langle X,Y\bigr \rangle_{\mathfrak{g}}\;,\;\;\;\forall \;g\in G,\;X,Y\in \mathfrak{g}\;,
        $
        and invariant:
        \begin{equation}
    \label{invariantekillingformg}
             \bigl \langle [X,Y],Z\bigr \rangle_{\mathfrak{g}}=\bigl \langle X,[Y,Z]\bigr \rangle_{\mathfrak{g}}\;,\;\;\;\forall \;X,Y,Z\in \mathfrak{g}
        \end{equation}
        
        \item $ \langle , \rangle_{\mathfrak{h}}$ is $G$-invariant:
       $
            \bigl \langle g\triangleright ^{\prime \prime}\xi, g\triangleright ^{\prime \prime}\nu \bigr \rangle_{\mathfrak{h}}=\bigl \langle \xi ,\nu\bigr \rangle_{\mathfrak{h}}\;,\;\;\;\forall \;g\in G,\;\xi ,\nu\in \mathfrak{h}\;,
       $
        and invariant:
        \begin{equation*}
             \bigl \langle [\xi ,\nu],\mu\bigr \rangle_{\mathfrak{h}}=\bigl \langle \xi ,[\nu,\mu]\bigr \rangle_{\mathfrak{h}}\;,\;\;\;\forall \;\xi ,\nu,\mu\in \mathfrak{h}\;.
        \end{equation*}
\end{enumerate}

\end{Definition}
\noindent Given the bilinear, non-degenerate, $G$-invariant, and invariant forms on $\mathfrak{g}$ and $\mathfrak{h}$, it is possible to define an antisymmetric bilinear function $\mathcal{T}:\mathfrak{h}\times \mathfrak{h}\rightarrow \mathfrak{g}$ by the rule
\begin{equation}
    \label{mapeoantisimetrico}
         \bigl \langle \mathcal{T}(\xi,\nu),X\bigr \rangle_{\mathfrak{g}}=-\bigl \langle \xi ,X\triangleright ^{\prime}\nu\bigr \rangle_{\mathfrak{h}}\;,\;\;\;\xi ,\nu\in \mathfrak{h},\;X\in \mathfrak{g},
    \end{equation}
    and a linear transformation $\partial ^{\;\prime}:\mathfrak{g}\rightarrow \mathfrak{h}$ defined by
$  \bigl \langle \partial ^{\;\prime}(X), u\bigr \rangle _{\mathfrak{h}}= \bigl \langle X, \partial _{\ast}(u)\bigr \rangle _{\mathfrak{g}},\;\forall \;X\in \mathfrak{g}, \;\forall \;u\in \mathfrak{h}.
   $ Let $\{\tau_{a}\}_{a=1,...,dim(G)}$ be a basis in $\mathfrak{g}$ and $\{e_{\alpha}\}_{\alpha=1,...,dim(H)}$ be a basis in $\mathfrak{h}$. Then, in components, it is possible to define:
\begin{itemize}
    \item The structure constants associated with the bases in the respective algebras given by the relations
    \begin{equation}
            [\tau _{a},\tau _{b}]_{\mathfrak{g}}=f_{ab}^{c}\tau _{c},\; \; \;    
[e_{\mu},e_{\nu}]_{\mathfrak{h}}=d_{\mu \nu}^{\alpha}\;e_{\alpha}.
        \end{equation}
        \item The matrices $(Q_{ab})=( \langle \tau _{a},\tau _{b} \rangle_{\mathfrak{g}})$ and $(q_{\mu,\nu})= (\langle e_{\mu}, e_{\nu}\rangle_{\mathfrak{h}})$ correspond to the $G$-invariant metrics in $\mathfrak{g}$ and $\mathfrak{h}$, respectively.
        \item The components of the action $(\triangleright ^{\prime})_{a \mu}^{\alpha}$ given by the relation
\begin{equation}
            \tau _{a}\triangleright ^{\prime} e_{\mu}=(\triangleright ^{\prime})_{a \mu}^{\alpha}e_{\alpha}.
        \end{equation}
        \item The components of $\mathcal{T}$ given by the relation
        \begin{equation}\label{bilinealantisimetricot}
          \mathcal{T}(e_{\mu}, e_{\nu})=\mathcal{T}_{\mu \nu}^{a}\;\tau_{a}\equiv(\triangleright ^{\prime})_{b \mu}^{\alpha}\;q_{\nu \alpha}\;Q^{ab}\;\tau_{a}.
        \end{equation}
\end{itemize}
\noindent From eq.~\eqref{invariantekillingformg}, we obtain the following relations between the structure constants and the non-degenerate bilinear form $Q$
        \begin{equation}\label{cosntantedeestructuraykillingform}   
            f_{ab}^{l}Q_{lc}=f_{bc}^{l}Q_{la}.
        \end{equation}
On the other hand, from the antisymmetry of $\mathcal{T}$, we have
\begin{equation}
        \label{relacionentrebanderitasykillingform}
            (\triangleright ^{\prime})_{b \mu}^{\alpha}\;q_{\nu \alpha}=-(\triangleright ^{\prime})_{b \nu}^{\alpha}\;q_{\mu \alpha}.
        \end{equation}
\begin{Example}
    Let $G$ be a Lie group, $H=\mathbb{R}^{n}$, and $\rho:G\rightarrow GL_{n}(\mathbb{R})$ a representation of $G$ on $\mathbb{R}^{n}$. Then $(G, H, \partial, \triangleright _{\rho})$ is a crossed module, where the action of $G$ on $H$, $\triangleright _{\rho}$, is defined via the representation as: $\bm {g \triangleright e=\rho (g)e}$, and the group morphism $\bm{\partial:H\rightarrow G}$ is the trivial morphism. The associated differential crossed module is $(\mathfrak{g}, \mathbb{R}^{n}, \partial _{\ast}, \triangleright ^{\prime})$, where the action of $\mathfrak{g}$ on $\mathbb{R}^{n}$, $\triangleright ^{\prime}$, is given by $\bm{\tau \triangleright ^{\prime}e=\rho _{\ast}(\tau)e}$, for any $\tau \in \mathfrak{g},\;e\in \mathbb{R}^{n}$ and the algebra morphism $\partial _{\ast}:\mathbb{R}^{n}\rightarrow \mathfrak{g}$ is the zero morphism, that is, $\bm{\partial _{\ast}=0}$. Finally, $\bm{g\triangleright ^{\prime \prime} e=\triangleright ^{\prime \prime}(g)e=\rho (g)e}$ for any $g\in G$, $e\in \mathbb{R}^{n}$.
\end{Example}

\begin{Example}
Let $G$ be a Lie group, $H =G$, $\partial = \text{id}_{G}$, and $\triangleright = \text{Ad}$ the conjugation automorphism. Then $(G,\; G,\; \text{id}_{G},\; \text{Ad})$ is a crossed Lie module with differential crossed module $(\mathfrak{g},\; \mathfrak{g},\; id_{\mathfrak{g}},\; ad)$ where $ad=[,]$ is the commuter in the Lie algebra.
\end{Example}
\noindent This $2$-group is trivializable and is known as the inner automorphism 2-group of $G$ which is denoted $INN(G)$~\cite{roberts2007inner}. It has the propriety that for every ordered pair of $1$-arrows there is exactly one $2$-arrow (codiscrete property). This follows from interpreting a $2$-arrow $\alpha :g \Rightarrow g^{\prime}$ as an orderer pair $(g,h)\in G\times H$ where $g^{\prime}=\partial (h)g$. Since $\partial =id_{G}$, it follows that $h=g^{\prime}g^{-1}$.
\begin{Example}
    \label{modulocruzadoabeliano}
       Let $G=SO(n)$, $H=Spin(n)$ the spin group,  $\partial= \rho$ the double cover function and $\triangleright_{\rho}$ the conjugation given by
       \begin{equation*}
           g\triangleright _{\rho} h:=g^{\prime}hg^{\prime-1},\;with \;\rho(g^{\prime})=g.
       \end{equation*}
       This is well-defined since $\rho$ is surjective and $Ker(\rho)=\{\pm id\}$. Then $(SO(n),Spin(n),\rho,\triangleright _{\rho})$ is a crossed Lie module but it is not codiscrete.\\
       A particular case of interest of this spin $2$-group is when $n=2$ where we have $(SO(2),U(1),\rho,\triangleright _{\rho})$ where
       \begin{equation*}
           \rho(e^{\frac{i}{2}\theta})=
           \begin{pmatrix}
            cos(\theta) & -sin(\theta)\\
            Sin(\theta) & cos(\theta)
            \end{pmatrix},
       \end{equation*}
       and $g\;\triangleright _{\rho}h=h$ for all $g\in SO(2)$. The associated differential crossed module is $(\mathfrak{so(2)}, \mathfrak{u(1)}, \rho_{\ast}, \triangleright ^{\prime})$, where 
       \begin{equation*}
           \rho _{\ast}(\frac{i}{2}\theta)=\theta J,\;\;\;J=
           \begin{pmatrix}
            0 & -1\\
            1 & 0
            \end{pmatrix},
       \end{equation*}
       is a isomorphism and $\triangleright ^{\prime}=0$.  Additionally,
        \begin{equation*}
            \langle \frac{i}{2},\frac{i}{2} \rangle_{\mathfrak{u(1)}}=\langle J,J\ \rangle_{\mathfrak{so(2)}}=-2,
        \end{equation*}
        which are bilinear, symmetric, non-degenerate, invariant under $U(1)$ and $SO(2)$, respectively, and invariant forms. Finally,
        \begin{equation*}
            \mathcal{T}=0,\;\;\; \partial ^{\prime}=\rho_{\ast}^{-1}.
        \end{equation*}
\end{Example}

\section{Differential forms valued in the Lie algebra
}
\label{formasdiferencialesvaluadas}
In this appendix, we will discuss the concepts of differential forms valued in a Lie algebra, as well as wedge products that can be defined given a differential crossed module.\\

\noindent Let $\{\tau_{a}\}_{a=1,...,\text{dim}(G)}$ be a basis in $\mathfrak{g}$ and $\{e_{\alpha}\}_{\alpha=1,...,\text{dim}(H)}$ be a basis in $\mathfrak{h}$. A $\mathfrak{g}$-valued $p$-form $A$ is an element of the set $\Omega ^{p}(M,\mathfrak{g})=\Omega ^{p}(M)\otimes \mathfrak{g}$, over the manifold $M$. Similarly, a $\mathfrak{h}$-valued $p$-form $\beta$ is an element of the set $\Omega ^{p}(M,\mathfrak{h})=\Omega ^{p}(M)\otimes \mathfrak{h}$, over the manifold $M$. Locally, they can be written as
\begin{equation}
        A\equiv A^{a} \tau_{a}=\frac{1}{p!}A_{\mu _{1},...,\mu _{p}}^{a}\tau_{a}dx^{\mu _{1}}\wedge ... \wedge dx^{\mu _{p}},
    \end{equation}
    and
    \begin{equation}
        \beta \equiv \beta^{\alpha} e_{\alpha}=\frac{1}{p!}\beta_{\mu _{1},...,\mu _{p}}^{\alpha}e_{\alpha}dx^{\mu _{1}}\wedge ... \wedge dx^{\mu _{p }},
    \end{equation}
For simplicity, the tensor product symbol has been omitted. The exterior derivative and the Hodge star operator act only on the part $\Omega ^{p}(M)$,
\begin{equation}
        \begin{split}
           & dA=(dA^{a})\tau _{a},\;\;\;d\beta=(d\beta ^{\alpha})e_{\alpha},\\
           & \star A=(\star A^{a})\tau _{a} ,\;\;\; \star \beta=(\star \beta ^{\alpha})e_{\alpha}. 
        \end{split}
    \end{equation}
    \noindent The commutator of a $\mathfrak{g}$-valued $p$-form $A\equiv A^{a} \tau_{a}$ and a $\mathfrak{g}$-valued $q$-form $B\equiv B^{b} \tau_{b}$ is
 \begin{equation}
       A\wedge ^{ad} B\equiv [A\wedge B]=(A^{a}\wedge B^{b})[\tau _{a},\tau _{b}],
    \end{equation}
    where
    $$            ad:\mathfrak{g}\times \mathfrak{g}  \rightarrow \mathfrak{g},\;\;\;
            (X,Y)  \rightarrow [X,Y],
       $$
    it is the adjoint representation of the Lie algebra $\mathfrak{g}$. In general, let $U$, $V$, and $W$ be vector spaces, and let
    \begin{equation*}
        T:U\times V \rightarrow W,
    \end{equation*}
be a bilinear transformation. Given a $p$-form $U$-valued $A\equiv A^{a}u_{a}$ and a $q$-form $V$-valued $B\equiv B^{\mu}v_{\mu}$, where $\{u_{a}\}$ is a basis of $U$ and $\{v_{\mu}\}$ is a basis of $V$, we can construct a $(p+q)$-form $W$-valued as   
\begin{equation}
        A\wedge ^{T} B=A^{a}\wedge B^{\mu}\;T(u_{a},v_{\mu}).
    \end{equation}
    This construction is independent of the choice of bases. So that, given a differential crossed module $(\mathfrak{g},\mathfrak{h},\partial _{\ast}, \triangleright ^{\prime})$ with a bilinear, symmetric, non-degenerate, $G$-invariant and invariant form, it is possible to define the following wedge products:
\begin{itemize}
    \item For $\mathfrak{g}$-valued forms:
    \begin{equation*}
            A\wedge ^{ad_{\mathfrak{g}}} B\equiv [A\wedge B]=(A^{a}\wedge B^{b})[\tau _{a},\tau _{b}]_{\mathfrak{g}}\equiv A^{a}\wedge B^{b} f_{ab}^{c}\tau _{c}\;.
        \end{equation*}
        
        \begin{equation}
            A\wedge ^{ \langle , \rangle_{\mathfrak{g}}} B\equiv \bigl \langle A \wedge B\bigr \rangle_{\mathfrak{g}}=A^{a}\wedge B^{b} \langle \tau _{a}, \tau _{b}\rangle_{\mathfrak{g}}\equiv A^{a}\wedge B^{b} Q_{ab}\;.
        \end{equation}
        \item  For $\mathfrak{h}$-valued forms:
        \begin{equation*}
            \eta \wedge ^{ad_{\mathfrak{h}}} \beta \equiv [\eta \wedge \beta]=\eta ^{\mu} \wedge \beta ^{\nu}[e_{\mu},e_{\nu}]_{\mathfrak{h}}\equiv \eta ^{\mu} \wedge \beta ^{\nu} d_{\mu \nu}^{\alpha}e_{\alpha}\;.
        \end{equation*}
        \begin{equation}
            \eta \wedge ^{\langle , \rangle_{\mathfrak{h}}} \beta \equiv \bigl \langle \eta \wedge \beta \bigr \rangle_{\mathfrak{h}} =\eta ^{\mu} \wedge \beta ^{\nu} \langle e_{\mu},e_{\nu} \rangle_{\mathfrak{h}}\equiv \eta ^{\mu} \wedge \beta ^{\nu} q_{\mu \nu}.
        \end{equation}

        \begin{equation}
        \label{productocuñaantisimetrica}
              \eta \wedge ^{\mathcal{T}} \beta=\eta ^{\mu} \wedge \beta ^{\nu} \mathcal{T}(e_{\mu},e_{\nu})\equiv \eta ^{\mu} \wedge \beta ^{\nu}\mathcal{T}_{\mu \nu}^{a}\tau _{a},
        \end{equation}
        where $\mathcal{T}:\mathfrak{h}\times \mathfrak{h} \rightarrow \mathfrak{g}$ is the antisymmetric bilinear function defined in \ref{mapeoantisimetrico}.
        \item For a $\mathfrak{g}$-valued form $A$ and a $\mathfrak{h}$-valued form $\beta$,
        \begin{equation}
        \label{productocuñaconbanderitas}
            A\wedge ^{\triangleright ^{\prime}} \beta=A^{a}\wedge \beta ^{\mu} (\tau _{a}\triangleright ^{\prime}e_{\mu})\equiv A^{a}\wedge \beta ^{\mu} (\triangleright ^{\prime})_{a \mu}^{\alpha}e_{\alpha}.
        \end{equation}
\end{itemize}

\noindent Given a $\mathfrak{g}$-valued 1-form $A$, it is possible to define another wedge product denoted as
\begin{equation}
         A\wedge A,
    \end{equation}
    by,
    \begin{equation*}
        A \wedge A (X_{1},X_{2})=[A (X_{1}),A (X_{2})],
    \end{equation*}
    for any pair of vector fields $(X_{1}, X_{2})$.
    
\section{Some Identities}
\label{app:a}
In this appendix, we will provide a detailed development of the variations of certain parts of the $BF$ action and the $BFCG$ action, which play a central role throughout the paper. This appendix will serve as a reference for readers interested in exploring in depth the specific variations of the $BF$ and $BFCG$ actions.

\noindent Let $M$ be an $n$-dimensional manifold, $B$ a $(n-2)$-form $\mathfrak{g}$-valued $G$-equivariant (as in Theorem \ref{teoderivadacovarianteequivariante}), $A$ the local 1-form connection and $F$ the 2-curvature. The following equality is useful for obtaining the variation of the $BF$ theory in the $n$-dimensional space $M$
\newpage
{\allowdisplaybreaks
\begin{align}\label{identidadBdeltaF}
         \int_M \bigl \langle B\wedge \delta F\bigr \rangle & = \int_M  B^{a}\wedge \delta F^{i}Q_{ia}\nonumber\\[0.2cm]
        & =\int_M  B^{a}\wedge \delta dA^{i}Q_{ia}+\frac{1}{2}B^{a}\wedge \delta (A^{b}\wedge A^{c})f_{bc}^{i}Q_{ia} \nonumber\\[0.2cm]
        & =\int_M  B^{a}\wedge \delta dA^{i}Q_{ia}+B^{a}\wedge A^{b}\wedge \delta A^{c}f_{bc}^{i}Q_{ia}\nonumber\\[0.2cm]
        & =\int_M  B^{a}\wedge \delta dA^{i}Q_{ia}+B^{a}\wedge A^{b}\wedge \delta A^{c}f_{ab}^{i}Q_{ic}\nonumber\\[0.2cm]
        & =\int_M  B^{a}\wedge \delta dA^{i}Q_{ia}+(-1)^{(n-1)}A^{a}\wedge B^{b}\wedge \delta A^{c}f_{ab}^{i}Q_{ic}\nonumber\\[0.2cm]
        & =\int_M  B^{a}\wedge \delta dA^{i}Q_{ia}+(-1)^{(n-1)}(d_{A}B-dB)^{a}\wedge \delta A^{i}Q_{ia}\nonumber\\[0.2cm]
        & =\int_M (-1)^{n}d(B^{a}\wedge \delta A^{i}Q_{ia})+(-1)^{(n-1)}d_{A}B^{a}\wedge \delta A^{i}Q_{ia}\nonumber\\[0.2cm]
        & =\int_M (-1)^{n}d\bigl \langle B\wedge \delta A\bigr \rangle +\bigl \langle \delta A\wedge d_{A}B\bigr \rangle.
    \end{align}}
 where $Q_{ia}$ are the components of the internal metric in some base $\{g_{a}\}$ of the Lie algebra of the Lie group $G$, $f^{i}_{bc}$ are the structure constants in this basis. We also employ the identity 
$ f_{ab}^{l}Q_{lc}=f_{bc}^{l}Q_{la}$, which ensures that there are no ambiguities in the wedge product $\langle \bullet \wedge \bullet \rangle$.\\
\noindent Now let $C$ be an $(n-3)$-form valued in $\mathfrak{h}$, $B$ an $\mathcal{G}$-equivariant (see Theorem \ref{teoderivadacovarianteequivariante}) $  (n-2)$-form valued in $\mathfrak{g}$,  $A$ the 1-connection valued in $\mathfrak{g}$, $\beta$ the 2-connection valued in $\mathfrak{h}$, and $\mathcal{F}$, $\mathcal{G}$, the fake curvature 2-form and curvature $3$-form, defined as
\begin{equation*}
        \mathcal{F}_{A,\beta}=F_{A}-\partial _{\ast}(\beta) \equiv dA+A\wedge A-\partial _{\ast}(\beta),\;\;\; \mathcal{G}_{A,\beta}=d\beta +A\wedge ^{\triangleright ^{\prime}}\beta.
    \end{equation*}
    The following equalities are useful for obtaining the variation of the $BFCG$ theory in the $n$-dimensional space $M$:
    \begin{equation}
        \begin{split}
             \int_M \bigl \langle B\wedge \delta \mathcal{F}\bigr \rangle _{\mathfrak{g}}& =\int_M \bigl \langle B\wedge \delta F\bigr \rangle _{\mathfrak{g}}-B^{a}\wedge \delta \partial _{\ast}(\beta)^{b}Q_{ab}\\[0.2cm]
             & =\int_M \bigl \langle B\wedge \delta F\bigr \rangle _{\mathfrak{g}}-B^{a}\wedge \delta \beta ^{\mu} \;(\partial_{\ast})_{\mu}^{b}\;Q_{ab}\\[0.2cm]
             & =\int_M \bigl \langle B\wedge \delta F\bigr \rangle _{\mathfrak{g}}-\delta \beta ^{\mu}\wedge B^{a} \;\partial_{a}^{\;\prime \;\alpha}\;q_{\mu \alpha}\\[0.2cm]
             & =\int_M (-1)^{n}d\bigl \langle B\wedge \delta A\bigr \rangle_{\mathfrak{g}} +\bigl \langle \delta A\wedge d_{A}B\bigr \rangle_{\mathfrak{g}}-\bigl \langle \delta \beta \wedge \partial ^{\;\prime}(B) \bigr \rangle_{\mathfrak{h}},
        \end{split}
    \end{equation}
    where $\partial ^{\;\prime}:\mathfrak{g}\rightarrow \mathfrak{h}$ is the linear transformation defined as
    \begin{equation*}
       \bigl \langle \partial ^{\;\prime}(X), u\bigr \rangle _{\mathfrak{h}}= \bigl \langle X, \partial _{\ast}(u)\bigr \rangle _{\mathfrak{g}}\;\;\;\;\forall \;X\in \mathfrak{g}, \;\forall \;u\in \mathfrak{h}.
    \end{equation*}
  From eqs. \eqref{bilinealantisimetricot}, \eqref{relacionentrebanderitasykillingform}, \eqref{productocuñaantisimetrica} and \eqref{productocuñaconbanderitas},
    \begin{equation}
        \begin{split}
            \int_M \bigl \langle C\wedge \delta \mathcal{G}\bigr \rangle_{\mathfrak{h}} & =\int_M C^{\mu}\wedge \delta \mathcal{G}_{\mu}\\[0.2cm]
            & =\int_M C^{\mu}\wedge \delta d\beta _{\mu}+ \left[C^{\mu}\wedge (\delta A \wedge ^{\triangleright ^{\prime}} \beta)_{\mu}+ C^{\mu}\wedge (A \wedge ^{\triangleright ^{\prime}} \delta \beta)_{\mu}\right]\\[0.2cm]
            & =\int_M \left[(-1)^{(n-1)}d(C^{\mu}\wedge \delta \beta _{\mu})+(-1)^{n}dC^{\mu}\wedge \delta \beta _{\mu}\right]\\[0.2cm]
            &\qquad {} \qquad {}\;+\left[C^{\mu}\wedge \delta A^{a} \wedge \beta ^{\nu}(\triangleright ^{\prime}) _{a\nu}^{\alpha}q_{\mu \alpha}+C^{\mu}\wedge A^{a} \wedge \delta \beta ^{\nu}(\triangleright ^{\prime}) _{a\nu}^{\alpha}q_{\mu \alpha}\right]\\[0.2cm]
            & = \int_M \left[(-1)^{(n-1)}d\bigl \langle C\wedge \delta \beta \bigr \rangle_{\mathfrak{h}}+(-1)^{n}\delta \beta^{\mu}\wedge  dC_{\mu}\right]\\[0.2cm]
            &\qquad {} \qquad {}\;+[(-1)^{(n-1)}\delta A^{a}\wedge C^{\mu}\wedge \beta ^{\nu}(\triangleright ^{\prime}) _{a\nu}^{\alpha}q_{\mu \alpha}\\[0.2cm]
             & \qquad {} \qquad {}\;+(-1)^{(n-1)}\delta \beta ^{\nu}\wedge A^{a}\wedge C^{\mu}(\triangleright ^{\prime})_{a\nu}^{\alpha}q_{\mu \alpha}]\\[0.2cm]
            & = \int_M \left[(-1)^{(n-1)}d\bigl \langle C\wedge \delta \beta \bigr \rangle_{\mathfrak{h}}+(-1)^{n}\delta \beta^{\mu}\wedge  dC_{\mu}\right]\\[0.2cm]
            &\qquad {} \qquad {}\;+[(-1)^{n}\delta A^{a}\wedge C^{\mu}\wedge \beta ^{\nu}\mathcal{T}_{\mu \nu}^{b}Q_{ab}\\[0.2cm]
            &\qquad {} \qquad {}\;+(-1)^{n}\delta \beta ^{\nu}\wedge A^{a}\wedge C^{\mu}(\triangleright ^{\prime})_{a\mu}^{\alpha}q_{\nu \alpha}]\\[0.2cm]
            & =\int_M (-1)^{(n-1)}d\bigl \langle C\wedge \delta \beta \bigr \rangle_{\mathfrak{h}}+(-1)^{n}\bigl \langle \delta \beta \wedge (dC+A\wedge ^{\triangleright ^{\prime}}C) \bigr \rangle_{\mathfrak{h}}\\[0.2cm]
            &\qquad {} \qquad {} \qquad {} \qquad {} \qquad {} \qquad {}+(-1)^{n}\bigl \langle \delta A\wedge (C\wedge ^{\mathcal{T}}\beta) \bigr \rangle_{\mathfrak{g}}.
            \end{split}
    \end{equation}  
    


\bibliographystyle{JHEP.bst}
\bibliography{biblio.bib}







 \end{document}